\documentclass[10pt]{article}%
\usepackage{amsmath}
\usepackage{graphicx}
\usepackage{amsfonts}
\usepackage{amssymb}
\usepackage{times}%
\providecommand{\U}[1]{\protect\rule{.1in}{.1in}}
\newtheorem{theorem}{Theorem}[section]
\newtheorem{lemma}[theorem]{Lemma}
\newtheorem{proposition}[theorem]{Proposition}
\newtheorem{definition}[theorem]{Definition}

\numberwithin{equation}{section}

\newenvironment{proof}[1][Proof]{\noindent\textbf{#1.} }{\ \rule{0.5em}{0.5em}}
\begin{document}

\begin{center}
.

\vspace{0.4in}

{\LARGE \textbf{Statistically Optimal Strategy  Analysis of a Competing Portfolio
Market with a Polyvariant Profit Function}}

\vspace{0.3in}

\smallskip\textsf{Bohdan Yu. Kyshakevych*, Anatoliy \ K. Prykarpatsky**, Denis
Blackmore***, Ivan P. Tverdokhlib****}

\medskip

\textsf{*) Department of Management, The Ivan Franko State Pedagogical
University,}

\textsf{Drohobych, Lviv Region, Ukraine}

\textsf{bogdan.kysh@gmail.com, peles@mail.ru}

\smallskip

\textsf{**) Department of Mining Geodesics, AGH-University of Science and
Technology,}

\textsf{Krakow, 30059, Poland and Department of Management, }

\textsf{The Ivan Franko State Pedagogical University,}

\textsf{Drohobych, Lviv Region, Ukraine}

\textsf{pryk.anat@ua.fm, prykanat@cybergal.com}

\smallskip

\textsf{***) Department of Mathematical Sciences and Center for Applied
Mathematics and }

\textsf{Statistics, New Jersey Institute of Technology, Newark, NJ 07102-1982}

\textsf{deblac@m.njit.edu}

\bigskip

\textsf{****) Department of Economics, The Ivan Franko National University, }

\textsf{Lviv 79001, Ukraine}

\textsf{i\_tverdok@franko.lviv.ua}
\end{center}

\vspace{0.2in}

\noindent\textbf{ABSTRACT: }A competing market model with a polyvariant profit
function that assumes \textquotedblleft zeitnot\textquotedblright\ stock
behavior of clients is formulated within the banking portfolio medium and then
analyzed from the perspective of devising statistically optimal strategies. An
associated Markov process method for finding an optimal choice strategy for
monovariant and bivariant profit functions is developed. Under certain
conditions on the bank "promotional" \ parameter with respect to the "fee"
\ for a missed share package transaction and at an asymptotically large enough
portfolio volume, universal transcendental equations - determining the optimal
share package choice among competing strategies with monovariant and bivariant
profit functions - are obtained.

\bigskip

\noindent\textbf{Keywords:} Zeitnot market modeling, statistically optimal
strategy, Markov process, asymptotic analysis

\medskip

\noindent\textbf{PACS Classification: }89.65.Gh

\section{Introduction}

Ever since the pioneering work of Markowitz \cite{mark}, the statistical
analysis of numerous types of portfolio markets - especially from the
perspective of formulating optimal choice strategies - has become an
increasingly active area of research. Moreover, many of the fruits of this
research have been adopted and standardized in a variety of influential
financial treatises such as Berezovsky \& Gnedin \cite{Y-1} and Brealy \&
Myers \cite{BM}. With the recent successes of financial mathematics and the
keen interest (some would say obsession) fueled by the uncertainty and
volatility of current economic markets, it is not surprising that there has
been something of an explosion in statistically optimal market strategy papers
such as Blanchet-Scalliet et al. \cite{BEJM}, Bronshtein \& Zav'yalova
\cite{BZ}, Chan \& Yung \cite{CY}, Davis et al. \cite{a-6, a-7}, Gollier
\cite{Goll}, Kyshakevych et al. \cite{Y-1a, Y-1b}, Maslov \cite{Mas}, Okui
\cite{Okui}, Reidel \cite{R}, Sun et al. \cite{MIT}, and Ye \& Peng \cite{YP},
to name just a few. Here we adapt and extend some of the techniques developed
in \cite{Y-1a, Y-1b} in order to add another piece to the puzzle of optimal
strategy formulation: one that treats \textquotedblleft
zeitnot\textquotedblright\ markets with polyvariant profit functions. The
zeitnot (not enough time) assumption imposes a strong time-horizon dependence
on our model, and establishes a certain commonality with the work in
\cite{BEJM}, but our work also has some striking differences with this and the
other research appearing in the literature.

{It is a well known that stock markets within the banking medium have a
regulative influence on a country's economic well being. This medium may have
within its portfolio large share packages of diverse business-industrial
structures, ordered by means of some natural indices of their
financial-economic attractiveness or worth to a potential client-buyer. For
modern zeitnot stock market processes, both at the fixed time constraint and
bounded access to the full resource information about share financial-economic
value, an optimal choice strategy \cite{a-6, Y-1,P-K,K-N,....} , identifying the most
desirable share package from a particular bank portfolio, assumes a great deal
of importance for clients.}

{The situation becomes much more complicated when many client-buyers are in
competition, and then a nontrivial fast choice problem arises subject to the
most worth share package within the portfolio. For, as was already mentioned
above, the "zeitnot" \ market character of such share market operations
provides a client with only comparative information data about their worth
during the choice process. Namely, if a client-buyer chooses some share
package from the bank portfolio, he or she can after learning its basic
characteristics buy it right away, or return the request back to the portfolio
and pass ahead to become familiar with a next share package. If its worth
characteristic proves to be equal or lower than those previously considered,
the client-buyer will right away pass on to choosing a next share packages
until he or she finds a share package with a worth characteristic higher than
all those considered previously. In this case the client-buyer should make a
decision as to whether this package is potentially the most valuable among all
the possible choices and stops the process by purchasing it. If the
client-buyer decides not to buy this share package, then he or she should
proceed to analyzing the worth characteristics of the next packages, taking
into account that the portfolio volume is finite and the market time is fixed.
}

{If there are two or more clients-buyers, a similar choice strategy subject to
the most valuable share package is followed, and based on an analysis of the
relative characteristics, both decide to buy the package, then the
client-buyer who acts fastest will acquire the package and be most successful.
The edge in speed will go to the client-buyer able to evaluate the potential
share package in the fewest number of steps. At the same time, the choice
process for identifying the most valuable share package is definitely affected
by certain additional Financial constraints, which essentially influence the
number of steps-requests to the portfolio data base. So, in a "zeitnot"
\ market, a client buyer ought to be charged a progressive amount of money
(fee) when using the request procedure subject to the portfolio data base for
each share package considered and then returned to the portfolio share
package. If, at last, the client-buyer stops at some potentially most valuable
share package (from his or her point of view) and buys it, the bank, as a
financial promotional-active institution, reimburses some money (gift) for the
successful commercial operation, thereby stimulating clients to engage in
active cooperation with the bank. }

{The competing stock market model within a bank portfolio medium under the
"zeitnot" \ market scheme delineated above, which governs the relationships
among clients-buyers, represents a fairly typical situation \cite{a-6,a-7,
Y-1} in a modern financial-economical context. As the whole choice process of
the most valuable share package tends to be quite casual and unstructured, it
is natural to employ stochastic process theory in its description. More
specifically, we shall employ certain aspects of minimax optimization
strategies and stopping rules associated with stochastic processes. A major
component of our analysis will be the construction of a mathematical model
capable of accurately reflecting the most important aspects of the stock
market processes described above. Once this is obtained as it is in the
sequel, we can employ fairly standard mathematical techniques to make
predictions of market behavior and formulate optimal strategies.}

{When analyzing optimal strategies subject to a competing stock market
portfolio model within a bank medium, an important problem arises for the
"zeitnot" \ market choice problem for two and more clients-buyers of share
packages, parameterized by a certain profit function. }

{In the first approximation, we assume that customers do not have financial
restrictions and have sufficient capital to purchase any stake. Under this
condition, when there are several variants of the profit function distributed
independently within the portfolio, its analysis is important for modeling the
optimal behavior of the clients-buyers, and thus, for ensuring stability of
the financial-economic processes. }

{In the proposed investigation we develop a method of using an associated
Markov-processes for the construction of the optimal strategy for the behavior
of two clients-buyers from the competing portfolio model described above. The
model is assumed to have either a given distributed mono- or bivariant profit
function governing \cite{Y-1,Y-Ar,Y-12,Y-13} the share packages in the bank
environment. }

\section{{Elements of Optimal Stops Theory}}

{Let }$(\Omega,\mathcal{F},P)${\ be a probability space \cite{Y-8,Y-9,Y-10}
with probability measure }$P${\ on the }$\sigma${-algebra }$\mathcal{F}${\ of
subsets of }$\Omega${, and }$x_{t}:\Omega\rightarrow\mathcal{H}${\ be a Markov
process for }$t\in\mathbb{Z}_{+}${\ or }$t\in\mathbb{R}_{+}${. For simplicity
we shall suppose that }$\mathcal{H}${\ is a discrete or finite-dimensional
topological space. Let us assume that on the space }$\mathcal{H}${\ there are
defined two functions: }$f:\mathcal{H}\rightarrow\mathbb{R}${, which is
interpreted as a revenue in the case that the process stops, and
}$c:\mathcal{H}\rightarrow\mathbb{R}${, which is interpreted as a fee for the
next process monitoring. Thus, if an agent monitors the trajectory of the
process }$x_{t}:\Omega\rightarrow\mathcal{H}${\ at the moments }$t=0,\ldots
,n${, and when }$\tau=n\in\mathbb{Z}_{+}${\ decides to stop monitoring, then
the revenue is the following: }%
\begin{equation}
f(x_{n})-\sum_{j=0}^{n-1}c(x_{j}). \label{Y1}%
\end{equation}
{Since the value (\ref{Y1}) is random, we should consider its mathematical
expectation and find a time }$\tau:=\tau^{\ast}\in\mathbb{Z}_{+}${\ when the
following equality }%
\begin{equation}
\tau^{\ast}=\arg\sup_{\tau}\{f(x_{\tau})-\sum_{j=0}^{\tau-1}c(x_{j})\}
\label{Y2}%
\end{equation}
{holds. Let us describe the concept of the \emph{stop moment}\ }$\tau
\in\mathbb{Z}_{+}${, which is a \emph{strategy}\ \cite{Y-1,Y-9,Y-13} of the
monitoring agent. In order to do this let us define a nondecreasing sequence
of the }$\sigma${-algebras }$\mathcal{F}_{n}:=\sigma\{x_{1},x_{2}%
,...,x_{n}\},n\in\mathbb{Z}_{+}${\ on the probability space }$(\Omega
,\mathcal{F},P)${, where }$\mathcal{F}_{n}${\ is a minimal }$\sigma${-algebra
containing all possible sets of the form }$\{\omega:x_{i}(\omega)\in B,0\leq
i\leq n\}${, where }$B\subset\mathcal{H}${\ is an arbitrary Borel set. }

\begin{definition}
{\label{Df_Y1.1} The Markov moment is a random quantity }$\tau=\tau(\omega)${
whose value lies in the set }$\mathbb{Z}_{+}$,{\ and for any }$n\in
\mathbb{Z}_{+}${\ }%
\begin{equation}
\{\omega:\tau(\omega)=n\}\in\mathcal{\mathcal{F}}_{n}. \label{Y3}%
\end{equation}

\end{definition}

{This condition means that decision about the end of monitoring at the moment
of time }$n\in\mathbb{Z}_{+}${\ is based only on the results of the monitoring
}$\{x_{1},x_{2},...,x_{n}\}${\ of the Markov process up to the moment }%
$n\in\mathbb{Z}_{+}${\ inclusive. }

\begin{definition}
{\label{Df_Y1.2} The Markov moment }$\tau\in\mathbb{Z}_{+}${ for which the
condition }$P\{\tau<\infty\}=1${\ is satisfied, or when the event }%
$\{\omega\in\Omega:t<\tau(\omega)\}\subset\mathcal{F}_{t}${\ for all }%
$t\in\mathbb{Z}_{+}${, is referred to as a \emph{stop moment}. }
\end{definition}

\begin{definition}
{\label{Df_Y1.3} The value }%
\begin{equation}
V(x_{0}):=\sup_{\tau}E_{x_{0}}\{f(x_{0})-\sum_{j=0}^{\tau-1}c(x_{j})\}
\label{Y4}%
\end{equation}
{is called the \emph{price} of the optimal stop problem. }
\end{definition}

{Let us consider the situation when the fee function }$c:\mathcal{H}%
\rightarrow\mathbb{R}${\ of the Markov process monitoring }$x_{t}%
:\Omega\rightarrow\mathcal{H}${\ with the transformation matrix }%
$\mathcal{P}:=\{p_{ij}:0\leq i,j\leq N\}${\ is zero, where we took into
account that }$card\,\mathcal{H}=N+1${. For convenience we shall suppose that
the function }$f:\mathcal{H}\rightarrow\mathbb{R}${\ is non-negative on
}$\mathcal{H}${. In order to simplify the analysis of the value (\ref{Y4}),
let us also introduce the so called }\emph{revaluation coefficient\ }%
$\alpha\in(0,1]${, which includes the cost of monitoring changes in time.
Then, if the agent-observer uses the Markov moment }$\tau\in\mathbb{Z}_{+}%
${\ as a strategy, the price of the optimal stop is }%
\begin{equation}
V(x_{0}):=\sup_{\tau}E_{x_{0}}\{\alpha^{\tau}f(x_{\tau})\}, \label{Y5}%
\end{equation}
{since the function }$c=0${. Now we define the operation }%
\begin{equation}
(\mathcal{P}f)_{i}:=\sum_{j=0}^{N}p_{ij}f(j), \label{Y6}%
\end{equation}
{for }$i=0,\ldots,N${, where by definition }$f(j):=f(x_{n_{j}}=j)${\ for the
some }$n_{j}\in\mathbb{Z}_{+},j\in\mathcal{H}\simeq\{0,1,...,N\}${. }

{The next definitions are useful for understanding material to be introduced
in the sequel:}

\begin{definition}
{\label{Df_Y1.4} The function }$g:\mathcal{H}\rightarrow\mathbb{R}_{+}${\ is
called an \emph{excessive function}, if }%
\begin{equation}
g(x)\geq\alpha(\mathcal{P}g)(x) \label{Y7}%
\end{equation}
{for all }$x\in\mathcal{H}${. }
\end{definition}

\begin{definition}
{\label{Df_Y1.5} An excessive function }$g:\mathcal{H}\rightarrow
\mathbb{R}_{+}${\ is called an \emph{excessive majorant} of the function
}$f:\mathcal{H}\rightarrow\mathbb{R}_{+}${, if }%
\begin{equation}
g(x)\geq f(x) \label{Y8}%
\end{equation}
{for all }$x\in\mathcal{H}${. The following lemma \cite{Y-1,Y-13} holds. }
\end{definition}

\begin{lemma}
{\label{Lm_Y1.6} If }$g:\mathcal{H}\rightarrow\mathbb{R}_{+}${\ is excessive
function and }$\tau\in\mathbb{Z}_{+}${\ is Markov moment, then for }$\alpha
\in(0,1]${\ }%
\begin{equation}
g(x)\geq\alpha(\mathcal{P}g)(x) \label{Y9}%
\end{equation}
{for all }$x\in\mathcal{H}${. }
\end{lemma}

\begin{proof}
{Let }$h:=g-\alpha\mathcal{P}g${\ and }$\alpha\in(0,1)${. Then, obviously,
}$h(x)\geq0${\ for all }$x\in\mathcal{H}${, since the conclusion that follows
from the condition (\ref{Y7}) is that }%
\begin{equation}
g(x)\geq\alpha(\mathcal{P}g)(x). \label{Y10}%
\end{equation}
{Rewriting the expression }$h:=g-\alpha\mathcal{P}g${\ in the form }%
\begin{equation}
g:=h+\alpha\mathcal{P}g, \label{Y11}%
\end{equation}
{we readily compute that }%
\begin{equation}
g:=h+\alpha\mathcal{P}h+\alpha^{2}\mathcal{P}^{2}h+...+\alpha^{n}%
\mathcal{P}^{n}h+..., \label{Y12}%
\end{equation}
$n\in\mathbb{Z}_{+}${, and expansion (\ref{Y12}) is convergent under the
condition that }$\alpha\in(0,1)${. Besides, since the mathematical expectation
}%
\begin{equation}
E_{i}\{h(x_{n})\}=\sum_{j=0}^{N}\bar{p}_{ij}^{(n)}h(j), \label{Y13}%
\end{equation}
{where }$\mathcal{P}^{n}:=\{\bar{p}_{ij}^{(n)}:0\leq i,j\leq N\},${\ the
expression (\ref{Y11}) can be rewritten as }%
\begin{equation}
g(x)=E_{x}\{\sum_{n=0}^{\infty}\alpha^{n}h(x_{n})\}. \label{Y14}%
\end{equation}
{Now the mathematical expectation can be calculated as }%
\begin{equation}
E_{x}\{\alpha^{\tau}g(x_{\tau})\}=E_{x}\{\alpha^{\tau}E_{x_{\tau}}\{\sum
_{n=0}^{\infty}\alpha^{n}h(x_{n})\}\}=E_{x}\{\sum_{n=0}^{\infty}\alpha
^{\tau+n}h(x_{\tau+n})\}. \label{Y15}%
\end{equation}
{Now comparing (\ref{Y15}) and (\ref{Y14}), we conclude that }%
\begin{equation}
g(x)\geq\alpha(\mathcal{P}g)(x) \label{Y16}%
\end{equation}
{for all }$x\in\mathcal{H},$ \ $\alpha\in(0,1)${. Letting }$\alpha
\rightarrow1${\ in (\ref{Y16}), it is easy to show that inequality (\ref{Y9})
holds for }$\alpha=1${, which completes the proof. }
\end{proof}

{The next theorem \cite{Y-13} plays a crucial role in our further analysis. }

\begin{theorem}
{\label{Tm_Y1.7} The profit function (\ref{Y5}) is the smallest excessive
majorant of the function }$f:\mathcal{H}\rightarrow\mathbb{R}_{+}${. }
\end{theorem}

{\medskip}

\begin{proof}
Observe {that it follows from (\ref{Y5}) that the price }$V:\mathcal{H}%
\rightarrow\mathbb{R}_{+}${\ is the excessive majorant of the function
}$f:\mathcal{H}\rightarrow\mathbb{R}_{+}${. Indeed, since }$V(x)=\sup_{\tau
}E_{x}\{\alpha^{\tau}f(x_{\tau})\},${\ the Markov moment }$\tau_{\epsilon}%
\in\mathbb{Z}_{+}${\ exists for any }$\epsilon>0${, and }%
\begin{equation}
E_{x}\{\alpha^{\tau_{\epsilon}}f(x_{\tau_{\epsilon}}\}>V(x)-\epsilon,
\label{Y17}%
\end{equation}
{where the value }$x\in\mathcal{H}${\ is fixed. Since }$card\ \mathcal{H}%
=N+1${\ is a finite quantity, the inequality (\ref{Y17}) also holds for all
}$x\in\mathcal{H}${. Let us calculate the mathematical expectation }%
\begin{equation}
E_{x}\{\alpha^{\tau^{\prime}}f(\tau^{\prime})\}=\sum_{j=1}^{N}p_{x,j}%
E_{j}\{\alpha^{1+\tau_{\epsilon}}f(x_{\tau})\}\geq\alpha\sum_{j=1}^{N}%
p_{x,j}V(j)-\alpha\epsilon, \label{Y18}%
\end{equation}
{where }$\tau^{\prime}:=1+\tau_{\epsilon}\in\mathbb{Z}_{+}${. From the
inequality (\ref{Y18}) we infer that }%
\begin{equation}
V(x)\geq E_{x}\{\alpha^{\tau^{\prime}}f(\tau^{\prime})\}\geq\alpha
(\mathcal{P}V)(x)-\alpha\epsilon\label{Y19}%
\end{equation}
{for any }$\epsilon>0\mathbf{.}${\ Calculating the limit in (\ref{Y19}) for
}$\epsilon\rightarrow0,${\ we find that }$V(x)\geq\alpha(\mathcal{P}%
V)(x)${\ for all }$x\in\mathcal{H}${, which means that the excess of the
profit function is }$V:\mathcal{H}\rightarrow\mathbb{R}_{+}.${\ } {Now Lemma
\ref{Lm_Y2.1} implies that for any Markov moment }$\tau\in\mathbb{Z}_{+}${\ we
have the following inequalities: }%
\begin{equation}
g(x)\geq E_{x}\{\alpha^{\tau}g(x_{\tau})\}\geq E_{x}\{\alpha^{\tau}f(x_{\tau
})\} \label{Y20}%
\end{equation}
{for all }$x\in\mathcal{H}${. Calculating the supremum of (\ref{Y20}) for
}$\tau\in\mathbb{Z}_{+},${\ we obtain }%
\begin{equation}
g(x)\geq\sup_{\tau}E_{x}\{\alpha^{\tau}f(x_{\tau})\}=V(x)\nonumber
\end{equation}
{for all }$x\in\mathcal{H}.${\ Thus the proof is complete. }
\end{proof}

{To calculate the choice price }$V:\mathcal{H}\rightarrow\mathbb{R}_{+}${\ we
use a criterion that can be formulated as the following theorem. }

\begin{theorem}
{\label{Tm_Y1.8} The optimal choice price }$V:\mathcal{H}\rightarrow
\mathbb{R}_{+}${\ is the least solution of the equation }%
\begin{equation}
V(x)=\max\{f(x),\alpha(\mathcal{P}V)(x)\} \label{Y21}%
\end{equation}
{for all }$x\in\mathcal{H}${\ and }$\alpha\in(0,1]${. }
\end{theorem}

\begin{proof}
{Using the expression (\ref{Y21}), let us define the operator $Q_{\alpha}$
acting on a function }$y:\mathcal{H}\rightarrow\mathbb{R}_{+}${\ for }%
$x\in\mathcal{H}${: }
\begin{equation}
(Q_{\alpha}y)(x):=\max\{f(x),\alpha(\mathcal{P}y)(x)\}. \label{Y22}%
\end{equation}
{It is easy to see that the following inequalities }%
\begin{equation}
y(x)\leq(Q_{\alpha}y)(x)\leq...\leq(Q_{\alpha}^{n}y)(x) \label{Y23}%
\end{equation}
{hold for all }$x\in\mathcal{H}.${ We now consider the following expression }%
\begin{equation}
\tilde{V}(x):=\lim_{n\rightarrow\infty}(Q_{\alpha}^{n}f)(x) \label{Y24}%
\end{equation}
{and show that the function (\ref{Y24}) satisfies the equation (\ref{Y21}).
Indeed, since }%
\begin{equation}
(Q_{\alpha}^{n}f)(x)=\max\{f(x),\alpha(\mathcal{P}Q_{\alpha}^{n-1}f(x)\}
\label{Y25}%
\end{equation}
{for all }$n\in\mathbb{Z}_{+}${, taking the limit as }$n\rightarrow\infty
${\ yields the result (\ref{Y21}). As every solution of the equation
(\ref{Y21}) is an excessive majorant, the solution of (\ref{Y24}) is the same
function. It remains only to show that this function is the smallest excessive
minorant of the function }$f:\mathcal{H}\rightarrow\mathbb{R}_{+}${. Indeed,
owing to the inequality (\ref{Y21}) and the definition of }$Q_{\alpha}%
-${\ operation, it is easy to show that for all }$x\in\mathcal{H}$,{\ }%
\begin{equation}
\lim_{n\rightarrow\infty}(Q_{\alpha}^{n}g)(x)=g(x). \label{Y26}%
\end{equation}
{Since }$g(x)\geq f(x)${\ for }$x\in\mathcal{H}${, it follows that
}$(Q_{\alpha}^{n}g)(x)\geq(Q_{\alpha}^{n}f)(x)${\ too for all }$n\in
\mathbb{Z}_{+}${. Taking the limit of the last inequality as }$n\rightarrow
\infty${, we obtain }%
\begin{equation}
g(x)\geq\lim_{n\rightarrow\infty}(Q_{\alpha}^{n}f)(x)=\tilde{V}(x)
\end{equation}
{for all }$x\in\mathcal{H}${, and this means that the solution (\ref{Y24}) is
the smallest excessive minorant of the function }$f:\mathcal{H}\rightarrow
\mathbb{R}_{+}\mathbf{.}${\ Thus, }$V(x)=\tilde{V}(x),x\in\mathcal{H},${ so
the proof is complete. }
\end{proof}

{In virtue of the properties of the solution (\ref{Y24}) of the problem
(\ref{Y21}), the validity of the following theorem is demonstrated in
\cite{Y-9,Y-10,Y-13}. }

\begin{theorem}
{\label{Tm_Y1.9} The Markov moment }$\tau^{\ast}\in\mathcal{H},${\ defined by
the condition (\ref{Y2}) as the moment of the first hit of the process }%
$x_{t}:\Omega\rightarrow\mathcal{H},t\in\mathbb{Z}_{+},${\ into the set
}$\Gamma_{+}:=\{x\in\mathcal{H}:V(x)=f(x)\}\mathbf{,}${\ is optimal. Besides,
}$\lim_{n\rightarrow\infty}P_{x}\{\tau^{\ast}>n\}=0,${\ as well as }%
$E_{x}\{\alpha^{\tau_{n}}V(x_{\tau_{n}})\}=V(x),${\ where }$\tau_{n}%
:=\min(\tau^{\ast},n),n\in\mathbb{Z}_{+}${, for all }$x\in\mathcal{H}.${\ }
\end{theorem}

{Let us consider the case when the fee mapping }$c:\mathcal{H}\rightarrow
\mathbb{R}_{+}${\ is nonzero. Then the price function has the following form }%
\begin{equation}
V_{\alpha}(x):=\sup_{\tau}E_{x}\{\alpha^{\tau}f(x_{\tau})-\sum_{i=0}^{\tau
-1}\alpha^{i}c(x_{i})\}, \label{Y28}%
\end{equation}
{where }$x\in\mathcal{H},${\ }$\alpha\in(0,1].${\ Analogously to Definition
\ref{Df_Y1.4}, we can formulate the following \cite{Y-13} definition. }

\begin{definition}
{\label{Df_Y1.10} The function }$g:\mathcal{H}\rightarrow\mathbb{R}_{+}${\ is
called an excessive function, if }%
\begin{equation}
g(x)\geq\alpha(\mathcal{P}g)(x)-c(x) \label{Y29}%
\end{equation}
{for all }$x\in\mathcal{H}.${\ }
\end{definition}

{We define }%
\begin{equation}
f_{\alpha}(x):=E_{x}\{\sum_{i=o}^{\infty}\alpha^{i}c(x_{i})\}, \label{Y30}%
\end{equation}
{where }$\alpha\in(0,1)${\ and }$x\in\mathcal{H}.${\ Whence, it is clear that
the choice price (\ref{Y28}) has the following representation: }%
\begin{equation}
V_{\alpha}(x):=\sup_{\tau}E_{x}\{\alpha^{\tau}[f(x_{\tau})+f_{\alpha}(x_{\tau
})]\}-f_{\alpha}(x) \label{Y31}%
\end{equation}
{for all }$x\in\mathcal{H}.${\ This means that the problem of the optimal stop
with the nonzero price }$c:\mathcal{H}\rightarrow\mathbb{R}_{+}${\ gives rise
to the analogous problem with the zero price for the observation; that is, to
the problem }%
\begin{equation}
\bar{V}_{\alpha}(x)=\sup_{\tau}E_{x}\{\alpha^{\tau}[f(x)+f_{\alpha}(x_{\tau
})]\}, \label{Y32}%
\end{equation}
$x\in\mathcal{H},${\ which solves the equation }%
\begin{equation}
\bar{V}_{\alpha}(x)=\max\{f(x)+f_{\alpha}(x),\alpha(\mathcal{P}\bar{V}%
_{\alpha})(x)\}. \label{Y33}%
\end{equation}
{Then the corresponding stop Markov moment is that }$\tau^{\ast}\in
\mathcal{H}_{+}${\ of the first observation in the set }$\Gamma_{+}%
:=\{x\in\mathcal{H}:\bar{V}_{\alpha}(x)=f(x)+f_{\alpha}(x)\}${. It is obvious,
that }$\Gamma_{+}=\{x\in\mathcal{H}:V_{\alpha}(x)=f(x)\}${. The result,
formulated above, is valid only for the }$\alpha\in(0,1)${. In the case when
}$\alpha=1${\ the function }$f_{\alpha}:\mathcal{H}\rightarrow\mathbb{R}_{+}%
${\ cannot be defined, and we need to find an alternative solution of the
problem (\ref{Y28}). }

{Let us set }$\alpha=1${\ and consider the powers }$\mathcal{P}^{m}${\ of the
transition probabilities matrix as }$m\rightarrow\infty\mathbf{.}${\ Then, it
follows from the ergodic theorem of A.A.~Markov \cite{Y-1,Y-8} that the
following asymptotic equality }$\mathcal{P}^{m}={\mathcal{S}}+h^{m}%
\mathcal{R}^{(m)}${\ holds, where }$|h|<1${\ and for all }$m\in\mathbb{Z}_{+}%
${\ the quantity }$\sup_{m\in\mathbb{Z}_{+}}||\mathcal{R}^{(m)}||\leq\bar
{r}<\infty,${\ and the matrix }$\mathcal{S}\in\mathcal{H}om(\mathbb{R}^{N+1}%
)${\ has exactly \ the same }$(N+1\in\mathbb{Z}_{+}{)}$ {positive vector-rows
}$q^{\intercal}\in\mathbb{R}_{+}^{N+1}${\ of the limit probabilities. Thus,
when }$\alpha=1${\ and the choice strategy price }$\tau=n\in\mathbb{Z}_{+}${,
the choice price (\ref{Y28}) is }%
\begin{align}
E_{x}\{f(x_{\tau})-\sum_{i=0}^{\tau-1}c(x_{i})\}  &  =f(x_{n})-c({x}%
)-\sum_{j=1}^{N}p_{{x,{j}}}c_{j}-...\label{Y34}\\
-\sum_{j=1}^{N}p_{x,{j}}^{(n-1)}c_{j}  &  =f(x_{n})-n\left\langle
q,c\right\rangle -\sum_{j=1}^{N}r_{x,{j}}c_{j},\nonumber
\end{align}
{where }$\left\langle .\ ,\ .\right\rangle ${\ is the ordinary scalar product
in the space }$\mathbb{R}^{N+1}${. From (\ref{Y34}) we see that when
}$\left\langle q,c\right\rangle <0$,{\ the choice price can be made
arbitrarily large while still stopping the monitoring process. If
}$\left\langle q,c\right\rangle \geq0,${\ then the situation is opposite to
the previous one, and it can be shown \cite{Y-12} that the quantity
(\ref{Y30}) for }$\alpha\in(\alpha_{0},1)${\ and some }$\alpha_{0}\in
(0,1)${\ is limited and positive. Owing to the action of the operator
(\ref{Y22}), }$Q_{\beta_{1}}(x)\leq Q_{\beta_{2}}(x)${\ for all }%
$x\in\mathcal{H}${\ and }$\beta_{1}\leq\beta_{2}\in(0,1)${\ is monotonic,
there exists a sequence }$\{\alpha_{n}\in(0,1):n\in\mathbb{Z}_{+}\}${\ such
that }$\lim_{n\rightarrow\infty}\alpha_{n}=1${\ and }$\lim_{n\rightarrow
\infty}V_{\alpha_{n}}(x^{\ast})=f(x^{\ast})${\ for some state }$x^{\ast}%
\in\mathcal{H}.${\ Under the condition }$card\ \mathcal{H}=N+1<\infty$,{\ the
limit }$\lim_{n\rightarrow\infty}V_{\alpha_{n}}(x)=f(x)${\ exists for every
}$x\in\mathcal{H}.${\ Thus, the set }%
\begin{equation}
\Gamma_{+}:=\{x\in\mathcal{\mathcal{H}}:V(x)=f(x)\} \label{Y35}%
\end{equation}
{is not empty when }$\left\langle q,c\right\rangle \geq0${, which is
tantamount to the optimality of the strategy }$\tau^{\ast}\in\mathcal{H}${\ of
the first hitting the observation into the set }$\Gamma_{+}.${\ }

{To formulate the concluding proposition for the case }$\alpha=1$,{\ we
partition the phase space }$\mathcal{H}${\ of the Markov process states
}$x_{t}:\Omega\rightarrow\mathcal{H},t\in\mathbb{Z}_{+},${\ into the subsets
of the nonessential states }$\mathcal{H}_{0}${\ and the classes }%
$\{\mathcal{H}_{i}:1\leq i\leq m_{N}\}${\ of the essential states with
nontrivial transition probabilities. Then the every essential class
}$\mathcal{H}_{i}\subset\mathcal{H},1\leq i\leq m_{N},${\ corresponds to the
vector of boundary probabilities }$q_{i}\in\mathbb{R}^{N+1}${\ and vector
}$c_{i}\in\mathbb{R}^{N+1}${, for which one can verify the following result. }

\begin{proposition}
{\label{Tm_Y1.11} If for some }$i\in\{1,\ldots,m_{N}\}${\ the quantity
}$\left\langle q_{i},c_{i}\right\rangle \geq0,${\ then the moment }$\tau
^{\ast}\in\mathcal{H}${\ of the first hit into the set }$\Gamma_{+}${\ in the
form (\ref{Y35}) is an optimal strategy. Nonessential states }$x\in
\mathcal{H}_{0},${\ for which we can find at least one set }$\mathcal{H}%
_{l}\subset\mathcal{H},${\ and }$\left\langle q_{l},c_{l}\right\rangle
<0,${\ belong to the subset }$\mathcal{H}_{0}\backslash\Gamma_{+}.${\ }
\end{proposition}

{In the next section we shall consider the problem of the optimal choice of
the competing portfolio model of the share market with a mono-variant profit
function, where the price function is defined by a constructive method
together with the associated Markov process. }

\section{{Mathematical Market Model with a Monovariant Profit Function}}

We begin with a constructive formulation of the model.

\subsection{{Model description}}

{Let }$(\Omega,\mathcal{F},P)${\ define \cite{a-2, Y-1} a probability space,
where }$\Omega${\ is the set of the elementary events with a selected }%
$\sigma${-algebra }$\mathcal{F}${\ of its subsets with probability measure
}$P${, defined on the subsets of }$\mathcal{F}${. Suppose that on the space
}$\Omega${\ there is a discrete Markov \cite{a-8, Y-10} process }%
$x:\mathbb{Z}_{+}\times\Omega\rightarrow\mathcal{H}${\ with the values in some
topological space }$\mathcal{H}${. For all }$t\in\mathbb{Z}_{+}${\ the
quantity }$x_{t}(\omega)\in\mathcal{H}${\ is random, and the set }%
$\{x_{t}(\omega)\in\mathcal{H}:\ t\in\mathbb{Z}_{+}\}${\ forms the virtual
trajectory of the possible states of the process. }

{We suppose that there exists an increasing family of the }$\sigma${-algebras
}$\{\mathcal{F}_{t}\subset\mathcal{F}:t\in\mathbb{Z}_{+}\}${\ such that }%
\begin{equation}
\mathcal{\mathcal{F}}_{s}\subset\mathcal{\mathcal{F}}_{t}\subset
\mathcal{\mathcal{F}} \label{e2.1}%
\end{equation}
{for all }$t>s\in\mathbb{Z}_{+}${. Then the process }$x:\mathbb{Z}_{+}%
\times\Omega\rightarrow\mathcal{H}${\ is called the \emph{adapted process} to
the family }$\{\mathcal{F}_{t}\subset\mathcal{F}:t\in\mathbb{Z}_{+}\}${, if
the mapping }$x_{t}:\Omega\rightarrow\mathcal{H}${\ is }$\mathcal{F}_{t}%
${-measurable for every }$t\in\mathbb{Z}_{+}${. For the process }%
$x:\mathbb{Z}_{+}\times\Omega\rightarrow\mathcal{H}${\ we introduce the
important definition of the Markov stop moment \cite{Y-9,Y-10}. It is a
mapping }$\tau:\Omega\rightarrow\mathbb{Z}_{+}${ such that the event
}$\{\omega\in\Omega:t<\tau(\omega)\}\subset\mathcal{F}_{t}${\ for all }%
$t\in\mathbb{Z}_{+}${. }

{Now consider an arbitrary mapping }$f:\mathcal{H}\rightarrow\mathbb{R}${, and
find the mathematical expectation \cite{a-2,Y-10} of the process }%
$f(x_{t}):\Omega\rightarrow\mathbb{R}${\ regarding the }$\sigma${-algebra
}$\mathcal{F}_{s}\subset\mathcal{F},${\ which we denote as }$E_{s}%
(f(x_{t})):=E\{(f(x_{t})|\mathcal{F}_{s}\},t>s\in\mathbb{Z}_{+}${. Then, by
definition, }%
\begin{equation}
\int_{A\in\mathcal{F}_{s}}E_{s}(f(x_{t}))dP_{s}:=\int_{A\in
\mathcal{\mathcal{F}}_{s}\subset\mathcal{\mathcal{F}}}f(x_{t})dP \label{e2.2}%
\end{equation}
{for all subsets }$A\in\mathcal{F}_{s}${, where the measure }$dP_{s}${\ on
}$\mathcal{F}_{s}${\ is defined as an induced measure }$i_{s}^{\ast}dP${\ with
respect to the embedding mapping }$i_{s}:\mathcal{F}_{s}\rightarrow
\mathcal{F},s\in\mathbb{Z}_{+}${. If we define the mathematical expectation
}$E_{s}(x_{t})${\ of the process }$x_{t}:\Omega\rightarrow\mathcal{H}${\ for
}$t>s\in\mathbb{Z}_{+}${\ and find that }$E_{s}(x_{t})=x_{s},${\ then this
process is called \cite{a-2,Y-1} a \emph{martingale process}. }

{Let }$f:\mathcal{H}\rightarrow\mathbb{R}_{+}${\ be a mapping that
characterizes the degree of usefulness of the choice of the element }%
$x\in\mathcal{H},${\ which models the database of the share package of the
bank portfolio. Then the function }%
\begin{equation}
V(a):=sup_{\tau}E_{a}({f(x_{t})}), \label{e2.3}%
\end{equation}
{where the supremum is taken over all possible Markov stop moments of the
process }$x_{t}:\Omega\rightarrow\mathcal{H},t\in\mathbb{Z}_{+},${\ under the
condition that }$x_{0}=a\in\mathcal{H},${\ is called the \emph{price of the
problem} of the optimal stop of the probability process, and can serve as a
client-buyer's choice price of the most valuable share package from the bank
portfolio in the "zeitnot" \ stock market. For the competing model of the
stock market in the bank portfolio environment, we need to construct the
corresponding price function of the optimal choice \cite{a-00,Y-1} of the most
wanted share package for every client-buyer, using the "zeitnot" \ stock
conditions of this process. }

{For the sake of convenience we suppose that there are only two clients-buyers
competing with each other at the time when the choice of the most valuable
share package from the proposed portfolio with the finite number }%
$N\in\mathbb{Z}_{+}${\ of the elements is made. All the share packages }%
$A_{i},i=1,\ldots,N${\ will be enumerated in such a way that }%
\begin{equation}
W(A_{1})<W(A_{2})<...<W(A_{N}), \label{e2.3a}%
\end{equation}
{where }\{$W(A_{i}):1\leq i\leq N\}${\ are share package values whose specific
expression is not important. The probability space }$\Omega${\ obviously
consists of all possible permutations }$\omega:=\{\omega_{1},\omega
_{2},...,\omega_{N}\}${\ of the set of numbers }$\{1,2,...,N\},${\ and we
assume that all of them have the same probability, since under the "zeitnot"
\ stock market situation conditions preliminary information is not important.
Thus, we denote the process of making the choice of the share package }%
$\omega_{n},\ n=1,\ldots,N$ {by the client-buyers in the }$n${-time round as
}$X_{n}^{(p)}(\omega)=\omega_{n},\ p=1,2${. In addition, we also denote the
stop moments of the process, which will result in the largest values of the
mathematical expectations of the corresponding price functions of the share
packages choice, as }$\tau_{p}(\omega)\in\mathcal{H}:=\{0,1,2,...,N\},\ p=1,2$%
{. The choice process of the most desirable share package }$A_{N}${, which
implicitly has the number }$N,${\ is complicated by the fact, that after the
share packages }$(X_{1}^{(p)},X_{2}^{(p)},...,X_{n}^{(p)}),\ p=1,2,${\ are
chosen and returned to the portfolio }$n(\in\{1,\ldots,N\})${\ times, because
each client-buyer lacks the information about their true prescribed price
values, and can only see their relative placement in the choice process, that
is, }$X_{i}^{(p)}<X_{j}^{(p)},${\ if }$W(A_{i})<W(A_{j})\ i\neq j\leq
n,\ p=1,2.${\ Consequently, it is natural to introduce families of }$\sigma
${-algebras of the events }$\mathcal{F}_{n}^{(p)},n=1,\ldots,N,p=1,2,$%
{\ induced the events }$(X_{i}^{(p)}<X_{j}^{(p)},i\neq j\leq n):=\mathcal{F}%
_{n}^{(p)},${\ where }$\mathcal{F}_{1}^{(p)}:=\{\emptyset,\Omega
\},\ p=1,2,${\ and to define two sets of new characteristic random quantities,
taking into account the above competition process involving the choice of the
most valuable share package. Let the mathematical expectations }%
\begin{equation}
V_{\tau_{1}}^{(1)}(\tau_{2}):=c_{\alpha}E\{\chi{_{\{X_{\tau_{1}}%
^{(1)}=N,X_{\tau_{2}}^{(2)}\neq N\}}}+\chi_{\{X_{\tau_{1}}^{(1)}=N,X_{\tau
_{2}}^{(2)}=N,\tau_{1}<\tau_{2}\}}\}- \label{e2.4}%
\end{equation}%
\begin{equation}
-\alpha\sum_{k=1}^{\tau_{1}-1}(k/N^{2})E\{\chi{_{\{X_{k}^{(1)}\neq
N,X_{\tau_{2}}^{(2)}\neq N\}}+\chi_{\{X_{k}^{(1)}\neq N,X_{\tau_{2}}%
^{(2)}=N,k<\tau_{2}\}}\}},\nonumber
\end{equation}
{and }%
\begin{equation}
V_{\tau_{2}}^{(2)}(\tau_{1}):=c_{\alpha}E\{\chi{_{\{X_{\tau_{2}}%
^{(2)}\nonumber=N,X_{\tau_{1}}^{(1)}\neq N\}}+\chi_{\{X_{\tau_{2}}%
^{(2)}=N,X_{\tau_{1}}^{(1)}=N,\tau_{2}<\tau_{1}\}}\}}- \label{e2.4a}%
\end{equation}

\begin{equation}
-\alpha\sum_{k=1}^{\tau_{2}-1}(k/N^{2})E\{\chi{_{\{X_{k}^{(2)}\neq
N,X_{\tau_{1}}^{(1)}\neq N\}}+\chi_{\{X_{k}^{(2)}\neq N,X_{\tau_{1}}%
^{(1)}=N,k<\tau_{1}\}}\}}%
\end{equation}
{define the corresponding price functions of the choice process of the most
desirable share package for both client-buyers, where }$c_{\alpha}>0${\ is a
fixed parameter representing the \textquotedblright
promotional\textquotedblright\ bank encouragement for the client-buyer to
purchase the share package from the portfolio, }$\alpha\in(0,1)${\ is a
corresponding coefficient of the \textquotedblleft fee\textquotedblleft\ for
every refusal of purchase of the share packages, and }$\tau_{1},\tau_{2}%
\in\mathcal{H}${\ are the corresponding Markov stop moments of the processes.
Since the choice processes for every client-buyer are analogous, it suffices
to consider in detail only the first problem of choosing the most valuable
share package from the following two problems: }%
\begin{equation}
arg\sup_{\tau_{1}}V_{\tau_{1}}^{(1)}(\tau_{2})=\tau_{1}^{\ast},\ \ \ arg\sup
_{\tau_{2}}V_{\tau_{2}}^{(2)}(\tau_{1})=\tau_{2}^{\ast}. \label{e2.5}%
\end{equation}
{In order to the extremum problems (\ref{e2.5}) we shall use the method of the
associated Markov processes for the Markov stop moments of the choice process,
which we describe next. }

\subsection{{Associated Markov process}}

{Let us consider the following sequence of the price function of the choice of
the most valuable share package by the first client-buyer: }%
\begin{equation}
V_{n}^{(1)}(\tau_{2}):=c_{\alpha}(P\{X_{n}^{(1)}=N,X_{\tau_{2}}^{(2)}\neq
N\}+P\{X_{n}^{(1)}=N,X_{\tau_{2}}^{(2)}=N,n<\tau_{2}\})-\nonumber
\end{equation}%
\begin{equation}
-\alpha\sum_{k=1}^{n-1}(k/N^{2})(P\{X_{k}^{(1)}\neq N,X_{\tau_{2}}^{(2)}\neq
N\}+P\{X_{k}^{(1)}\neq N,X_{\tau_{2}}^{(2)}=N,k<\tau_{2}\}), \label{e3.1}%
\end{equation}
{where }$n=1,\ldots,\tau_{1},\alpha\in(0,1),c_{\alpha}>0,${\ and it is assumed
that the second client-buyer follows the optimal (so called \textquotedblright
threshold\textquotedblright) strategy with the Markov stop moment }$\tau
_{2}(l)>l${\ under the condition that Markov stop moment of the choice of the
first client-buyer is }$\tau_{1}(l)=l\in\mathcal{H}.${\ To add specificity to
the choice strategy of the most valuable shares package by the first
client-buyer, let us calculate the corresponding probabilities (\ref{e3.1})
taking into account the family of the associated }$\sigma${-algebras
}$\mathcal{F}_{n}^{(p)},n=1,\ldots,\tau_{1},p=1,2:${\ }%
\begin{align}
V_{n}^{(1)}(\tau_{2})  &  =c_{\alpha}P\{X_{n}^{(1)}=N|\mathcal{\mathcal{F}%
}_{n}^{(1)}\}[P\{X_{\tau_{2}}^{(2)}\neq N\}\label{e3.2}\\
+P\{X_{\tau_{2}}^{(2)}  &  =N,n<\tau_{2}\}]\nonumber\\
-\alpha\sum_{k=1}^{n-1}\frac{k}{N^{2}}P\{X_{k}^{(1)}  &  \neq N\}[P\{X_{\tau
_{2}}^{(2)}\neq N\}+P\{X_{\tau_{2}}^{(2)}=N,k<\tau_{2}\}]\nonumber
\end{align}

{It should be mentioned, that for }$n=1,\ldots,\tau_{1}${\ the conditional
probability }%
\begin{align}
P\{X_{n}^{(1)}  &  =N|\mathcal{\mathcal{F}}_{n}^{(1)}\}=P\{X_{n}^{(1)}%
=N:X_{n}^{(1)}>\max(X_{1}^{(1)},X_{2}^{(1)},..,X_{n-1}^{(1)})\}\nonumber\\
&  =P\{X_{n}^{(1)}=N\}/P\{X_{n}^{(1)}>\max(X_{1}^{(1)},X_{2}^{(1)}%
,..,X_{n-1}^{(1)})\}\nonumber\\
&  =\frac{1}{N}/(\frac{(n-1)!}{n!})1_{\{X_{n}^{(1)}>\max(X_{1}^{(1)}%
,X_{2}^{(1)},..,X_{n-1}^{(1)})\}}, \label{e3.3}%
\end{align}
{and for every }$k=1,\ldots,n${\ the conditional probability }%
\begin{align}
P\{X_{\tau_{2}}^{(2)}  &  =N,k<\tau_{2}\}+P\{X_{\tau_{2}}^{(2)}\neq
N,n<\tau_{2}\}\nonumber\\
&  =1-P\{X_{\tau_{2}}^{(2)}=N,\tau_{2}\leq k\}. \label{e3.4}%
\end{align}
{Thus, the price function of the choice (\ref{e3.2}) for the first
client-buyer for }$n=1,\ldots,\tau_{1}$ {has the following form: }%
\begin{equation}%
\begin{array}
[c]{c}%
V_{n}^{(1)}(\tau_{2})=\frac{c_{\alpha}n}{N}(1-P\{X_{\tau_{2}}^{(2)}=N,\tau
_{2}\leq n\}\\
\label{e3.5}-\frac{\alpha(N-1)}{N}\sum_{k=1}^{n-1}\frac{k}{N^{2}}%
(1-P\{X_{\tau_{2}}^{(2)}=N,\tau_{2}\leq k\}.
\end{array}
\end{equation}
{To calculate the probabilities }$P\{X_{\tau_{2}}^{(2)}=N,\tau_{2}\leq
k\},k=1,\ldots,n,${\ in the expression (\ref{e3.5}) we need to consider the
random sequences of the Markov stop moments associated with the process of
choosing the most valuable share package by the client-buyers: }%
\begin{equation}
x_{n}^{(p)}:=\min\{t>x_{n-1}^{(p)}:X_{t}>\max(X_{t-1},...,X_{1})\},
\label{e3.6}%
\end{equation}
{where }$x_{n}^{(p)}\in\mathcal{H}${\ is a moment of choice of the next
candidate for the most valuable share package by the corresponding
client-buyer. The random sequences (\ref{e3.6}) are figure definitively for
the price function (\ref{e3.5}), whose main properties are defined \cite{a-3}
by the following lemma. }

\begin{lemma}
{\label{Lm_e3.1} The sequences }$x_{n}^{(p)}\in\mathcal{H},n=1,\ldots
,N,p=1,2,${\ in the form (\ref{e3.6}) are discrete Markov chains on the phase
space }$\mathcal{H}${\ with the transition probabilities }%
\begin{equation}
p_{ij}=\left\{
\begin{array}
[c]{c}%
\frac{i}{j(j-1)},\text{ \ \ \ \ \ }0\leq i<j;\text{ \ \ \ \ }0,\ \ \ \ i\geq
j\geq0,\\
\\
\ \ \ 1,\text{ \ \ \ \ \ \ \ }i=0,\ j=1;\text{ \ \ \ }0,\text{ \ \ \ }%
i=0,\ j>1,\\
\frac{i}{N},\text{\ \ \ \ \ \ }j=0;\text{ \ \ \ \ \ \ \ \ \ }0,\ \ \ \ \text{
\ \ \ }i\geq j>0
\end{array}
\right.  \label{e3.7}%
\end{equation}
{for all }$0\leq i,j\leq N,${\ where the additional state }$\{0\}${\ of the
sequences break is added, which the process settles into after receiving the
most valuable share package. }
\end{lemma}

{Let us denote the optimal stop moments of the consequences (\ref{e3.6}) as
}$\hat{\tau_{p}}\in\mathcal{H},\ p=1,2${. Then the following relationships }%
\begin{equation}
\tau_{p}=x_{\hat{\tau}_{p}}, \label{e3.8}%
\end{equation}
{hold, where }$p=1,2.$ Now{ consider the arbitrary Markov sequence in the form
of (\ref{e3.6}) and the following decomposition of the phase space
}$\mathcal{H}${\ into the direct sum of the subspaces associated with the
sequence of price functions (\ref{e3.5}), which is }%
\begin{align}
\mathcal{\mathcal{H}}_{+}  &  :=\{j\in\mathcal{\mathcal{H}}:\mathcal{(P}%
V^{(1)}(\tau_{2}))_{j}>V_{j}^{(1)}(\tau_{2})\},\label{e3.9}\\
\mathcal{\mathcal{H}}_{-}  &  :=\{j\in\mathcal{\mathcal{H}}:\mathcal{(P}%
V^{(1)}(\tau_{2}))_{j}\leq V_{j}^{(1)}(\tau_{2})\},\nonumber
\end{align}
{where }$\mathcal{P}:=\{p_{ij}:0\leq i,j\leq N\}${\ is a matrix of the
transition probabilities (\ref{e3.7}). Then the following theorem \cite{Y-12}
obtains. }

\begin{theorem}
{\label{Th{e3.2}} Let the matrix }$\mathcal{P}${\ of the transition
probabilities (\ref{e3.7}) be such that }$p_{ij}=0${\ for all }$i\in
\mathcal{H}_{+}${\ and }$j\in\mathcal{H}_{-}${. Then the moment }$\hat{\tau
}_{1}\in\mathcal{H}${\ of the first entrance of the random sequence }%
$\{x_{n}^{(1))}:n=0,\ldots,N\}${\ into the set }$\mathcal{H}_{-}${\ is optimal
for the sequence of price function }$\{V_{n}^{(1)}(\tau_{2}):n=0,\ldots
,N\}.${\ }
\end{theorem}

{In order to apply theorem \ref{Th{e3.2}}, we calculate the probabilities
}$P\{X_{\tau_{2}}^{2}=N,\tau_{2}\leq k\}${\ in the (\ref{e3.5}) for all
}$0\leq k\leq N${\ under the condition that }$\tau_{1}(l)=x_{\hat{\tau}_{1}%
}^{(1)}(l):=l\in\mathcal{H}${. Then, if }$k=1,\ldots,l-1${, the probability }%
\begin{equation}
P\{X_{\tau_{2}}^{(2)}=N,\tau_{2}\leq k\}=P\{X_{\tau_{2}(l)}^{(2)}=N,\tau
_{2}(l)\leq k\}=0, \label{e3.10}%
\end{equation}
{since }$\tau_{2}(l)\geq l,${\ and if }$k=l,\ldots,N,${\ }%
\begin{align}
P\{X_{\tau_{2}(l)}^{(2)}  &  =N,\tau_{2}(l)\leq k\}=\sum\limits_{j=l}%
^{k}P\{X_{\tau_{2}(l)}^{(2)}=N,\tau_{2}(l)\leq j\}\nonumber\\
&  =\sum\limits_{j=l}^{k}P\{X_{\tau_{2}(l)}^{(2)}=N|\tau_{2}(l)=j\}P\{\tau
_{2}(l)=j\}\label{e3.11}\\
&  =\sum\limits_{j=l}^{k}P\{X_{j}^{(2)}=N:X_{j}^{(1)}>\max(X_{1}^{(2)}%
,X_{2}^{(2)},..,X_{j-1}^{(2)})\}P\{\tau_{2}(l)=j\}\nonumber\\
&  =\sum\limits_{j=l}^{k}P\{\tau_{2}(l)=j\}\frac{j}{N}.\nonumber
\end{align}
{In order to calculate the probability }$P\{\tau_{2}(l)=j:j\in\mathcal{H}\}${,
we note that it follows the direct Kolmogorov equation }%
\begin{equation}
P\{\tau_{2}(l)=j\}=\left\{
\begin{array}
[c]{c}%
1,\text{ \ \ \ \ \ \ \ \ \ \ \ \ \ \ \ \ \ \ \ \ \ \ \ }\ j=1,\\
\sum\limits_{i=1}^{j-1}P\{x_{\hat{\tau}_{2}(l)}^{(2)}=i\}p_{ij},\text{
}\ j=\overline{2,l-1},\text{\ \ \ \ \ }\\
\sum\limits_{i=1}^{l-1}P\{x_{\hat{\tau}_{2}(l)}^{(2)}=i\}p_{ji},\text{
\ \ \ }j=\overline{l,N},
\end{array}
\right.  \label{e3.12}%
\end{equation}
{\cite{a-6,Y-13} and (\ref{e3.12}) that }%
\begin{equation}
P\{\tau_{2}(l)=j\}=\left\{
\begin{array}
[c]{c}%
\frac{1}{j},\text{ \ \ \ \ \ \ \ \ }j=\overline{1,l-1},\\
\ \ \ \ \ \ \ \frac{l-1}{j(j-1)},\text{ \ \ \ \ \ \ \ \ \ }\ j=\overline
{l,N},\text{\ \ \ \ \ \ \ \ \ \ \ \ \ \ \ \ }%
\end{array}
\right.  \label{e3.13}%
\end{equation}
{\label{e3.14} From (\ref{e3.13}) and (\ref{e3.11}) we can find for
}$k=\overline{l,N}${, that }%
\begin{equation}
\{X_{\tau_{2}(l)}^{(2)}=N,\tau_{2}(l)\leq k\}=\sum\limits_{j=l}^{k}\frac
{l-1}{N(j-1)}.
\end{equation}
{Thus, substituting the result of (\ref{e3.14}) into (\ref{e3.15}), we can get
the final expression for the price function for the first client-buyer: }%
\begin{equation}%
\begin{array}
[c]{c}%
V_{n}^{(1)}(\tau_{2})=c_{\alpha}n(1-\frac{l-1}{N}\sum\limits_{j=l}^{n}\frac
{1}{j-1})-\frac{\alpha(N-1)}{N}\sum\limits_{k=1}^{n-1}\frac{k}{N^{2}}\\
-\frac{\alpha(N-1)}{N}\sum\limits_{k=l}^{l-1}\frac{k}{N^{2}}-\frac
{\alpha(N-1)}{N}\sum\limits_{k=l}^{n}\frac{k}{N^{2}}(1-\frac{l-1}{N}%
\sum\limits_{j=l}^{k}\frac{1}{j-1})\\
=c_{\alpha}n(1-\frac{l-1}{N}\sum\limits_{j=l}^{n}\frac{1}{j-1})-\frac
{\alpha(N-1)n(n+1)}{2N^{3}}+\frac{\alpha(N-1)(l-1)}{N^{2}}\sum\limits_{k=l}%
^{n}\frac{k}{N^{2}}\sum\limits_{j=l}^{k}\frac{1}{j-1}%
\end{array}
\label{e3.15}%
\end{equation}
{for all }$n=1,\ldots,N${. Now in order to solve the first equation in
(\ref{e2.5}) it is easy to calculate }$\tau_{1}^{\ast}=\arg V_{\tau_{1}}%
^{(1)}(\tau_{2})\in\mathcal{H}${\ using Theorem \ref{Th{e3.2}}. Thus, the
obtained sequence (\ref{e3.15}) of the optimal choice of the most valuable
shares package by the first client must be stopped at the moment }$\tau
_{1}(l)=l=x_{\hat{\tau}_{1}(l)}^{(1)}\in\mathcal{H}${, which we can find
solving the inequalities }%
\begin{align}
\mathcal{(P}V^{(1)}(\tau_{2}))_{l-1}  &  >V_{l-1}(\tau_{2}),\label{e3.16}\\
\mathcal{(P}V^{(1)}(\tau_{2}))_{l}  &  \leq V_{l}^{(1)}(\tau_{2}).\nonumber
\end{align}
{Let }$l\in\mathcal{H}${\ satisfy the inequalities (\ref{e3.16}). The the
following lemma is readily verified. }

\begin{lemma}
{\label{Lm_e3.3} Under the condition that promotional coefficient }$c_{\alpha
}\geq\alpha/2>0$,{\ the sequence (\ref{e3.15}) induces the decomposition of
the phase space }$\mathcal{H}${ with }%
\begin{equation}
\mathcal{\mathcal{H}}_{+}=\{1,\ldots,l-1\},\text{ \ \ \ \ \ \ }%
\mathcal{\mathcal{H}}_{-}=\{l,\ldots,N\}. \label{e3.17}%
\end{equation}

\end{lemma}

{It follows from Lemma \ref{Lm_e3.3} that }$\tau_{1}(l)=l\in\mathcal{H}${,
which satisfies the inequalities (\ref{e3.16}), and yields the optimal choice
strategy of the most valuable share package by the first client-buyer. It is
obvious from symmetry considerations that the competing choice problem
involving the behavior strategy of the second client-buyer must be the same. }

\subsection{{Asymptotic analysis}}

{The main equation of the choice process of the most valuable share package
for the optimal strategy (\ref{e3.16}) has the form: }%
\begin{align}
&  c_{\alpha}-\frac{\alpha(N-1)(l+1)}{2N^{3}}+\frac{\alpha(N-1)}{N^{2}%
}=\label{e4.1}\\
&  =c_{\alpha}(\sum\limits_{j=l-1}^{N-1}\frac{i}{j}-\frac{(l-1)}{N}%
\sum\limits_{j=l}^{N-1}\frac{1}{j}\sum\limits_{k=l-1}^{j}\frac{1}{k}%
)-\frac{\alpha l(N-1)}{2N^{3}}\sum\limits_{j=l+1}^{N}\frac{j+1}{j-1}%
+\nonumber\\
&  +\frac{\alpha(N-1)l(l-1)}{N^{2}}\sum\limits_{j=l+1}^{N}\frac{1}{j(j-1)}%
\sum\limits_{k=l}^{j}\frac{k}{N^{2}}\sum\limits_{j=l}^{k}\frac{1}%
{j-1}.\nonumber
\end{align}
{In order to simplify the analysis of the equation (\ref{e4.1}), we suppose
that the bank portfolio contains a large number }$N\in\mathbb{Z}_{+}${\ of
share packages. Thus, for the optimal choice strategy of the first
client-buyer the stop moment }$\tau_{1}(l):=l(N)\in\mathcal{H}${\ satisfies
asymptotic condition }$\lim_{N\rightarrow\infty}l(N)/N:=z\in(0,1)${. Taking
this into account, using asymptotic analysis \cite{Y-Ge,Y-14}, we find that
the relation (\ref{e4.1}) at }$N\rightarrow\infty${\ turns into the following
transcendental equation for finding the stop parameter }$z^{\ast}\in(0,1)${: }%
\begin{equation}
c_{\alpha}(1+\ln z+\frac{z}{2}\ln^{2}z)+\frac{\alpha}{2}z(1-z)=\frac{\alpha
}{2}z^{2}[lnz\frac{1}{2}(1-z)(3-z)]. \label{e4.2}%
\end{equation}
{The solution }$z^{\ast}\in(0,1)${\ depends heavily on the choice of a bank
\textquotedblleft gift\textquotedblright-parameter }$c_{\alpha}\in
\mathbb{R}_{+},${\ which is naturally limited by the positiveness of the price
function (\ref{e3.15}). Namely, it is easy to see that }%
\begin{equation}
c_{\alpha}-\alpha/2\geq0 \label{e4.3}%
\end{equation}
{must hold for every }$\alpha\in(0,1)${. If we assume the lowest risk
condition of losses of the bank shares seller, then the optimal choice is
}$c_{\alpha}=\alpha/2.${\ In this case equation (\ref{e4.2}) takes an
invariant form with respect to the interest rate of the \textquotedblleft
fee\textquotedblleft\ }$\alpha\in(0,1)${\ for purchases the potential desired
share package that have yet to be made by the client-buyer: }%
\begin{equation}
1+\ln z+\frac{z}{2}\ln^{2}z+z(1-z)=z^{2}[\ln z+\frac{1}{2}(1-z)(3-z)].
\label{e4.4}%
\end{equation}
{This transcendental equation (\ref{e4.4}) has the only one real solution
}$z^{\ast}\simeq0,21\in(0,1)${. Accordingly we can now formulate the next
behavior strategy as follows: When the number }$N\in\mathbb{Z}_{+}${\ of the
share packages in the bank portfolio is large enough, the optimal strategy of
the choice of the most valuable share package by the first client-buyer is to
compare the relative value of the first }$l=z^{\ast}N\in\mathbb{Z}_{+}%
${\ shares, and then to choose the first shares package whose value is greater
then all of those previously compared. }

\subsection{{Some conclusions}}

{Our portfolio competing share market model under the condition of "zeitnot"
\ stock choice of potentially the most valuable share package by client-buyers
appears to be a well known discrete Markov process on the phase space
}$\mathcal{H}=\{0,1,...,N\}${. As it has been shown, when the bank chooses the
most useful \textquotedblright promotional\textquotedblleft\ parameter
}$c_{\alpha}=\alpha/2\in(0,1)${, the client-buyer's optimal strategy choice of
the most valuable share package is defined by the universal transcendental
equation (\ref{e4.3}) independent the \textquotedblright fee\textquotedblleft%
-parameter }$\alpha\in(0,1)$ and{\ under the condition that the values of the
number of packages within the portfolio are large. }

{It should be noticed that our model is a somewhat simplified version of the
"zeitnot" \ stock behavior of clients/share buyers when they do not dispose of
a priori information about the qualitative characteristics of the portfolio.
Moreover, we assumed that every client-buyer possesses sufficient financial
capital for the purchase of any share package of the bank portfolio. }

{In the case if there exist either some financial constraints on clients funds
subject to portfolio share packages prices prescribed by a bank or several
quality parameters, the corresponding clients optimal behavior strategies are
essentially more complicated, and is a subject of analysis in the next
section. }

\section{{Mathematical Model of the Market with a Bivariant Profit Function}}

Our construction of the model with a bivariant profit function has some
similarities with the monovariant case, but there are some striking
differences as well.

\subsection{{Model description}}

{We take as a base the mathematical model of the bank share portfolio and the
process of client-buyer's choice of the share package described above and
developed in \cite{Y-1a}. Let us suppose that there are two competing
client-buyers in the process of choosing the most valuable share package with
a finite number }$N\in\mathbb{Z}_{+}${\ of elements. All share packages
}$A_{i},i=1,\ldots,N,${\ are a priori numbered in such a way that }%
\begin{equation}
W_{1}(A_{i})<W_{1}(A_{2})<...<W_{1}(A_{N}),\ \ \ W_{2}(A_{\sigma(1)}%
)<W_{2}(A_{\sigma(2)})<...<W_{2}(A_{\sigma(N)}), \label{Y1.1}%
\end{equation}
{where }$\{W_{i}(A_{j}):\ 1\leq j\leq N\},\ \ i=1,2${, are rankings of
usefulness characteristics of share packages, which are distributed
independently within a given portfolio; that is, the permutation }$\sigma\in
S_{N}${\ of the ordered set of numbers }$\{1,2,...,N\}${\ is random. The
probability space }$\Omega${\ consists of all possible pairs of permutations
}$\{\omega_{1},...,\omega_{N}\}\times\{\sigma(\omega_{1}),...,\sigma
(\omega_{N})\}${\ of the set of numbers }$\{1,2,...,N\}${, naturally assumed
to have equal probability. Thus, we denote the result of a client-buyer's
choice of the share packages }$A_{n}${, }$n=1,\ldots,N${, preceded by an }%
$n${-time examination as a }$\Omega_{n}^{(s)}:=(X_{n}^{(s)}(\omega
),(Y_{n}^{(s)}(\omega)))\in\{1,2,...,N^{(x)}:=N\}\times\{\sigma(1),\sigma
(2),..,\sigma(N^{(y)}:=N\},s=1,2${, and the Markov stop moments of the process
of choice of the most desired share package by client-buyers, under the
conditions that the values of mathematical expectations of the respective
choice price functions will be the largest, as }$\tau_{s}(\omega
)\in\mathcal{H}:=\{0,1,2,..,N\},s=1,2${. We choose the price function for the
first client-buyer in the following form: }%
\begin{equation}%
\begin{array}
[c]{c}%
V_{\tau_{1}}^{(1)}(\tau_{2})=c_{\alpha}[E\{\chi_{\{\Omega_{\tau_{1}}%
^{(1)}=(N^{(x)},N^{(y)}),\Omega_{\tau_{2}}^{(2)}\neq(N^{(x)},N^{(y)}%
)\vee\Omega_{\tau_{1}}^{(1)}=(\sigma^{-1}(N^{(x)}),N^{(y)}),\Omega_{\tau_{2}%
}^{(2)}\neq(\sigma^{-1}(N^{(x)}),N^{(y)})\}}\}+\\
+E\{\chi_{\{\Omega_{\tau_{1}}^{(1)}=(N^{(x)},N^{(y)}),\Omega_{\tau_{2}}%
^{(2)}=\newline(N^{(x)},N^{(y)}),\tau_{1}<\tau_{2}\vee\Omega_{\tau_{1}}%
^{(1)}\newline=(\sigma^{-1}(N^{(x)}),N^{(y)}),\Omega_{\tau_{2}}^{(2)}%
=(\sigma^{-1}(N^{(x)}),N^{(y)}),\tau_{1}<\tau_{2}\}}\}-\\
-\alpha\sum_{k=1}^{\tau_{1}-1}\frac{k}{N^{2}}[E\{\chi_{\{\Omega_{\tau_{1}%
}^{(1)}\neq(N^{(x)},N^{(y)}),\Omega_{\tau_{2}}^{(2)}\neq(N^{(x)},N^{(y)}%
)\vee\Omega_{k}^{(1)}\neq(\sigma^{-1}(N^{(x)}),N^{(y)}),\Omega_{\tau_{2}%
}^{(2)}\neq(\sigma^{-1}(N^{(x)}),N^{(y)})\}}\}+\\
+E\{\chi_{\{\Omega_{k}^{(1)}\neq(N^{(x)},N^{(y)}),\Omega_{\tau_{2}}^{(2)}%
\neq(N^{(x)},N^{(y)}),k<\tau_{2}\vee\Omega_{k}^{(1)}\neq(\sigma^{-1}%
(N^{x}),N^{y}),\Omega_{\tau_{2}}^{(2)}=(\sigma^{-1}(N^{(x)}),N^{(y)}%
),k<\tau_{2}\}}\}],
\end{array}
\label{Y1.2}%
\end{equation}
{where }$c_{\alpha}>0${\ is a corresponding bank gift-coefficient, and
}$\alpha>0${\ is a "fee"\ -coefficient for the unmade transaction of
purchase-sale of the shares package. The choice price function for the second
client is obtained in the same way. In order to calculate, for example, the
quantity }%
\begin{equation}
\tau_{1}^{\ast}:=\arg\sup_{\tau_{1}\in\mathcal{\mathcal{H}}}V_{\tau_{1}}%
^{(1)}(\tau_{2}), \label{Y1.3}%
\end{equation}
{which characterize the most optimal share package choice strategy of the
first client-buyer, we need to construct \cite{Y-1,Y-1a,Y-13} the basic
associated Markov sequences }%
\begin{equation}
x_{n+1}^{(s)}:=\min\{t>x_{n}^{(s)}:X_{t}^{(s)}>\max(X_{t-1}^{(s)}%
,...,X_{1}^{(s)})\vee(Y_{t}^{(s)}>\max(Y_{t-1}^{(s)},...,Y_{1}^{(s)})\},
\label{Y1.4}%
\end{equation}
{where the quantities }$x_{n}^{(s)}\in\mathcal{H},n=1,\ldots,N,s=1,2,${\ are
the moments of the most valuable shares package for the corresponding
clients-buyers. The Markov sequences (\ref{Y1.4}) are characterized
\cite{Y-1,Y-12,Y-13} by the following lemma. }

\begin{lemma}
{\label{Lm_Y2.1} The integer sequences (\ref{Y1.4}) are the discrete Markov
chains on the phase space }$\mathcal{H}${\ with the transition probabilities }%
\[
p_{ij}^{(s)}=\left\{
\begin{array}
[c]{c}%
\frac{\lbrack2j(j-1)-i]i^{2}}{j^{2}(j-1)^{2}(2i-1)},\ 1\leq
i<j;\ \ \ 0,\ \ \ i\geq j\geq0;\\
1,\ \ \ i=0,\ \ j=1;\ \ \ 0,\ \ \ i=0,j>1,\ \ \ \ \ 0;\ \ \\
1-\sum_{k=i+1}^{N}\frac{[2k(k-1)-i]i^{2}}{k^{2}(k-1)^{2}(2i-1)},\ \ \ j=0,
\end{array}
\right.
\]%
\begin{equation}
\label{Y1.5}
\end{equation}
{for }$s=1,2${\ and }$i,j\in\mathcal{H}.${\ }
\end{lemma}

{Thus, we have constructed two Markov sequences (\ref{Y1.4}) associated with
the most valuable share package choice process by means of which we can
calculate the quantity (\ref{Y1.3}), using the following result \cite{Y-12} as
the criterion. }

\begin{theorem}
{\label{Tm_Y2.2} Let the matrix }$\mathcal{P}:=\{p_{ij}^{(1)}:i,j\in
\mathcal{H}\}${\ of the transition probabilities be such that }$p_{ij}%
^{(1)}=0${\ for all }$i\in\mathcal{H}_{+},j\in\mathcal{H}_{-}${, where }%
\begin{align}
\mathcal{\mathcal{H}}_{+}  &  :=\{j\in\mathcal{\mathcal{H}}:(\mathcal{P}%
V^{(1)}(\tau_{2}))_{j}>V_{j}^{(1)}(\tau_{2})\},\nonumber\label{Y1.6}\\
\mathcal{\mathcal{H}}_{-}  &  :=\{j\in\mathcal{\mathcal{H}}:(\mathcal{P}%
V^{(1)}(\tau_{2}))_{j}\leq V_{j}^{(1)}(\tau_{2})\}.
\end{align}
{Then the Markov sequence (\ref{Y1.4}) for optimal choice of the most valuable
share package by the first client-buyer can be broken at the moment }$\tau
_{1}(l)=l=x_{\hat{\tau}_{1}(l)}^{(1)}\in\mathcal{H}${, which can be found
solving the inequalities (\ref{Y1.6}). }
\end{theorem}

{The corresponding choice price function of the share package in (\ref{Y1.6})
is given as }%
\begin{align}
V_{n}^{(1)}(\tau_{2})  &  =c_{\alpha}[P\{X_{n}^{(1)}=N^{(x)},X_{\tau_{2}%
}^{(2)}\neq N^{(x)}\vee Y_{n}^{(1)}=N^{(y)},Y_{\tau_{2}}^{(2)}\neq
N^{(y)}\}+\nonumber\label{Y1.7}\\
+P\{X_{n}^{(1)}  &  =N^{(x)},X_{\tau_{2}}^{(2)}=N^{(x)},n<\tau_{2}\vee
Y_{n}^{(1)}=N^{(y)},Y^{(2)}\neq N^{(y)}\}+\nonumber\\
+P\{X_{n}^{(1)}  &  =N^{(x)},X_{\tau_{2}}^{(2)}\neq N^{(x)}\vee Y_{n}%
^{(1)}=N^{(y)},Y_{\tau_{2}}^{(2)}=N^{(y)},n<\tau_{2}\}+\nonumber\\
+P\{X_{n}^{(1)}  &  =N^{(x)},X_{\tau_{2}}^{(2)}=N^{(x)},n<\tau_{2}\vee
Y_{n}^{(1)}=N^{(y)},Y_{\tau_{2}}^{(2)}=N^{(y)},n<\tau_{2}\}-\nonumber\\
-\alpha\sum\limits_{k=1}^{n-1}\frac{k}{N^{2}}[P\{X_{k}^{(1)}  &  \neq
N^{(x)},X_{\tau_{2}}^{(2)}\neq N^{(x)}\wedge Y_{k}^{(1)}\neq N^{(y)}%
,Y_{\tau_{2}}^{(2)}\neq N^{(y)},n<\tau_{2}\}+\\
+P\{X_{k}^{(1)}  &  \neq N^{(x)},X_{\tau_{2}}^{(2)}\neq N^{(x)}\wedge
Y_{k}^{(1)}\neq N^{(y)},Y_{\tau_{2}}^{(2)}=N^{(y)},k<\tau_{2}\}+\nonumber\\
+P\{X_{k}^{(1)}  &  \neq N^{(x)},X_{\tau_{2}}^{(2)}=N^{(x)},k<\tau_{2}\wedge
Y_{k}^{(1)}\neq N^{(y)},Y_{\tau_{2}}^{(2)}\neq N^{(y)}\}+\nonumber\\
+P\{X_{k}^{(1)}  &  \neq N^{(x)},X_{\tau_{2}}^{(2)}=N^{(x)},k<\tau_{2}\wedge
Y_{k}^{(1)}\neq N^{(y)},Y_{\tau_{2}}^{(2)}=N^{(y)},k<\tau_{2}\}],\nonumber
\end{align}
{where the bank gift-parameter }$c_{\alpha}>0${\ is chosen from the condition
}$V_{n}^{(1)}(\tau_{2})>0${\ for all }$n=1,\ldots,N.${\ Thus, after
calculating the value of the function of price of choice (\ref{Y1.7}) of the
most valuable share package by the first client-buyer by means of Theorem
\ref{Tm_Y2.2}, the structure of the sets }$\mathcal{H}_{+}${\ and
}$\mathcal{H}_{-}${\ on the transition probabilities (\ref{Y1.5}) needs to be
analyzed, as we do in the next subsection. }

\subsection{{Associated Markov process and structural analysis of the model}}

{Taking into account the structure of the independent family of associated
}$\sigma-${\ algebras }$\{\mathcal{F}_{n}^{(s)},1\leq n\leq\tau_{1}%
\},s=1,2,${\ let us rewrite expression (\ref{Y1.7}) in the following form: }%
\begin{align}
V_{n}^{(1)}(\tau_{2})  &  =c_{\alpha}[P\{X_{n}^{(1)}=N^{(x)},X_{\tau_{2}%
}^{(2)}\neq N^{(x)}\}+P\{Y_{n}^{(1)}=N^{(y)},Y_{\tau_{2}}^{(2)}\neq
N^{(y)}\}-\nonumber\\
-P\{X_{n}^{(1)}  &  =N^{(x)},X_{\tau_{2}}^{(2)}\neq N^{(x)}\}P\{Y_{n}%
^{(1)}=N^{(y)},Y_{\tau_{2}}^{(2)}\neq N^{(y)}\}+\nonumber\\
+P\{X_{n}^{(1)}  &  =N^{(x)},X_{\tau_{2}}^{(2)}=N^{(x)},n<\tau_{2}%
\}+P\{Y_{n}^{(1)}=N^{(y)},Y^{(2)}\neq N^{(y)}\}-\nonumber\\
-P\{X_{n}^{(1)}  &  =N^{(x)},X_{\tau_{2}}^{(2)}=N^{(x)},n<\tau_{2}%
\}P\{Y_{n}^{(1)}=N^{(y)},Y^{(2)}\neq N^{(y)}\}+\nonumber\\
+P\{X_{n}^{(1)}  &  =N^{(x)},X_{\tau_{2}}^{(2)}\neq N^{(x)}\}+P\{Y_{n}%
^{(1)}=N^{(y)},Y_{\tau_{2}}^{(2)}=N^{(y)},n<\tau_{2}\}-\nonumber\\
-P\{X_{n}^{(1)}  &  =N^{(x)},X_{\tau_{2}}^{(2)}\neq N^{(x)}\}P\{Y_{n}%
^{(1)}=N^{(y)},Y_{\tau_{2}}^{(2)}=N^{(y)},n<\tau_{2}\}+\nonumber\\
+P\{X_{n}^{(1)}  &  =N^{(x)},X_{\tau_{2}}^{(2)}=N^{(x)},n<\tau_{2}%
\}+P\{Y_{n}^{(1)}=N^{(y)},Y_{\tau_{2}}^{(2)}=N^{(y)},n<\tau_{2}\}-\nonumber\\
-P\{X_{n}^{(1)}  &  =N^{(x)},X_{\tau_{2}}^{(2)}=N^{(x)},n<\tau_{2}%
\}P\{Y_{n}^{(1)}=N^{(y)},Y_{\tau_{2}}^{(2)}=N^{(y)},n<\tau_{2}\}-\nonumber\\
-\alpha\sum\limits_{k=1}^{n-1}\frac{k}{N^{2}}[P\{X_{k}^{(1)}  &  \neq
N^{(x)},X_{\tau_{2}}^{(2)}\neq N^{(x)}\}P\{Y_{k}^{(1)}\neq N^{(y)},Y_{\tau
_{2}}^{(2)}\neq N^{(y)}\}+\label{Y2.1}\\
+P\{X_{k}^{(1)}  &  \neq N^{(x)},X_{\tau_{2}}^{(2)}\neq N^{(x)}\}P\{Y_{k}%
^{(1)}\neq N^{(y)},Y_{\tau_{2}}^{(2)}=N^{(y)},k<\tau_{2}\}+\nonumber\\
+P\{X_{k}^{(1)}  &  \neq N^{(x)},X_{\tau_{2}}^{(2)}=N^{(x)},k<\tau
_{2}\}P\{Y_{k}^{(1)}\neq N^{(y)},Y_{\tau_{2}}^{(2)}\neq N^{(y)}\}+\nonumber\\
+P\{X_{k}^{(1)}  &  \neq N^{(x)},X_{\tau_{2}}^{(2)}=N^{(x)},k<\tau
_{2}\}P\{Y_{k}^{(1)}\neq N^{(y)},Y_{\tau_{2}}^{(2)}=N^{(y)},k<\tau
_{2}\}],\nonumber
\end{align}
{where we use the fact that the respective traces of both observations of the
values of usefulness are distributed independently. It follows from the
results of \cite{Y-1a} that (\ref{Y2.1}) can be rewritten as }%
\begin{align}
V_{n}^{(1)}(\tau_{2})  &  =c_{\alpha}\{2P\{X_{n}^{(1)}=N^{(x)}%
|\mathcal{\mathcal{F}}_{n}^{(1)}\vee\mathcal{\mathcal{F}}_{n}^{(2)}%
\}[P\{X_{\tau_{2}}^{(2)}\neq N^{(x)}\}+\nonumber\\
+P\{X_{\tau_{2}}^{(2)}  &  =N^{(x)},n<\tau_{2}\}]+2P\{Y_{n}^{(1)}%
=N^{(y)}|\mathcal{\mathcal{F}}_{n}^{(1)}\vee\mathcal{\mathcal{F}}_{n}%
^{(2)}\}\times\nonumber\\
\times\lbrack P\{Y_{\tau_{2}}^{(2)}  &  \neq N^{(y)}\}+P\{Y_{\tau_{2}}%
^{(2)}=N^{(y)},n<\tau_{2}\}]-\nonumber\\
-P\{X_{n}^{(1)}  &  =N^{(x)}|\mathcal{\mathcal{F}}_{n}^{(1)}\vee
\mathcal{\mathcal{F}}_{n}^{(2)}\}P\{Y_{n}^{(1)}=N^{(y)}|\mathcal{\mathcal{F}%
}_{n}^{(1)}\vee\mathcal{\mathcal{F}}_{n}^{(2)}\}\times\nonumber\\
\times\lbrack P\{X_{\tau_{2}}^{(2)}  &  \neq N^{(x)}\}+P\{X_{\tau_{2}}%
^{(2)}=N^{(x)},n<\tau_{2}\}]\times\nonumber\\
\times\lbrack P\{Y_{\tau_{2}}^{(2)}  &  \neq N^{(y)}\}+P\{Y_{\tau_{2}}%
^{(2)}=N^{(y)},n<\tau_{2}\}]\}-\\
-\alpha\sum\limits_{k=1}^{n-1}\frac{k}{N^{2}}P\{X_{k}^{(1)}  &  \neq
N^{(x)}\}P\{Y_{k}^{(1)}\neq N^{(y)}\}[P\{X_{\tau_{2}}^{(2)}\neq N^{(x)}%
)+\nonumber\\
+P\{X_{k}^{(1)}  &  \neq N^{(x)},X_{\tau_{2}}^{(2)}=N^{(x)},k<\tau
_{2}\}]\times\nonumber\\
\times\lbrack P\{Y_{\tau_{2}}^{(2)}  &  \neq N^{(y)}\}+P\{Y_{\tau_{2}}%
^{(2)}=N^{(y)},k<\tau_{2}\}],\nonumber
\end{align}
{or the equivalent form }%
\begin{align}
V_{n}^{(1)}(\tau_{2})  &  =c_{\alpha}\{2P\{X_{n}^{(1)}=N^{(x)}%
|\mathcal{\mathcal{F}}_{n}^{(1)}\vee\mathcal{\mathcal{F}}_{n}^{(2)}%
\}[P\{X_{\tau_{2}}^{(2)}\neq N^{(x)}\}+\nonumber\\
+P\{X_{\tau_{2}}^{(2)}  &  =N^{(x)},n<\tau_{2}\}]+2P\{Y_{n}^{(1)}%
=N^{(y)}|\mathcal{\mathcal{F}}_{n}^{(1)}\vee\mathcal{\mathcal{F}}_{n}%
^{(2)}\}\times\nonumber\\
\times\lbrack P\{Y_{\tau_{2}}^{(2)}  &  \neq N^{(y)}\}+P\{Y_{\tau_{2}}%
^{(2)}=N^{(y)},n<\tau_{2}\}]-\nonumber\\
-P\{X_{n}^{(1)}  &  =N^{(x)}|\mathcal{\mathcal{F}}_{n}^{(1)}\vee
\mathcal{\mathcal{F}}_{n}^{(2)}\}P\{Y_{n}^{(1)}=N^{(y)}|\mathcal{\mathcal{F}%
}_{n}^{(1)}\vee\mathcal{\mathcal{F}}_{n}^{(2)}\}\times\nonumber\\
\times\lbrack P\{X_{\tau_{2}}^{(2)}  &  \neq N^{(x)}\}+P\{X_{\tau_{2}}%
^{(2)}=N^{(x)},n<\tau_{2}\}]\times\nonumber\\
\times\lbrack P\{Y_{\tau_{2}}^{(2)}  &  \neq N^{(y)}\}+P\{Y_{\tau_{2}}%
^{(2)}=N^{(y)},n<\tau_{2}\}]\}-\label{Y2.2}\\
-\alpha\sum\limits_{k=1}^{n-1}\frac{k}{N^{2}}P\{X_{k}^{(1)}  &  \neq
N^{(x)}\}P\{Y_{k}^{(1)}\neq N^{(y)}\}[P\{X_{\tau_{2}}^{(2)}\neq N^{(x)}%
)+\nonumber\\
+P\{X_{k}^{(1)}  &  \neq N^{(x)},X_{\tau_{2}}^{(2)}=N^{(x)},k<\tau
_{2}\}]\times\nonumber\\
\times\lbrack P\{Y_{\tau_{2}}^{(2)}  &  \neq N^{(y)}\}+P\{Y_{\tau_{2}}%
^{(2)}=N^{(y)},k<\tau_{2}\}].\nonumber
\end{align}
{It should be remarked, that for }$n=1,\ldots,\tau_{1}${the conditional
probabilities }%
\begin{align}
P\{X_{n}^{(1)}  &  =N^{(x)}|\mathcal{\mathcal{F}}_{n}^{(1)}\vee
\mathcal{\mathcal{F}}_{n}^{(2)}\}=P\{X_{n}^{(1)}:X_{n}^{(1)}>\max
(X_{n-1}^{(1)},...,X_{1}^{(1)})\}=\nonumber\label{Y2.3}\\
&  =P\{X_{n}^{(1)}=N^{(x)}\}/P\{X_{n}^{(1)}>\max(X_{n-1}^{(1)},...,X_{1}%
^{(1)})\}=\nonumber\\
&  =\frac{1}{N}/(\frac{(n-1)!}{n!})=\frac{n}{N}1_{\{X_{n}^{(1)}>\max
(X_{n-1}^{(1)},...,X_{1}^{(1)})\}},
\end{align}
and {analogously, }%
\begin{equation}
P\{Y_{n}^{(1)}=N^{(x)}|\mathcal{\mathcal{F}}_{n}^{(1)}\vee\mathcal{\mathcal{F}%
}_{n}^{(2)}\}=\frac{n}{N}1_{\{X_{n}^{(1)}>\max(X_{n-1}^{(1)},...,X_{1}%
^{(1)})\}}. \label{Y2.4}%
\end{equation}
{It also is easy to compute that for every }$k=1,\ldots,n${\ the conditional
probabilities }%
\begin{align}
P\{X_{\tau_{2}}^{(2)}  &  \neq N^{(x)}\}+P\{X_{\tau_{2}}^{(2)}=N^{(x)}%
,k<\tau_{2}\}=1-P\{X_{\tau_{2}}^{(2)}=N^{(x)},\tau_{2}\leq k\},\nonumber\\
P\{Y_{\tau_{2}}^{(2)\ddot{}}  &  \neq N^{(y)}\}+P\{Y_{\tau_{2}}^{(2)}%
=N^{(y)},k<\tau_{2}\}=1-P\{X_{\tau_{2}}^{(2)}=N^{(x)},\tau_{2}\leq k\}.
\label{Y2.5}%
\end{align}
{Thus, substituting the expressions (\ref{Y2.3})-(\ref{Y2.5}) in (\ref{Y2.2})
for all }$n=1,\ldots,\tau_{1},${\ we find that }%
\begin{align}
V_{n}^{(1)}(\tau_{2})  &  =c_{\alpha}[\frac{2n}{N}(1-P\{X_{\tau_{2}}%
^{(2)}=N^{(x)},\tau_{2}\leq n\})\nonumber\\
+\frac{2n}{N}(1-P\{Y_{\tau_{2}}^{(2)}  &  =N^{(y)},\tau_{2}\leq
n\})-\label{Y2.6}\\
-\frac{n^{2}}{N^{2}}(1-P\{X_{\tau_{2}}^{(2)}  &  =N^{(x)},\tau_{2}\leq
n\})(1-P\{Y_{\tau_{2}}^{(2)}=N^{(y)},\tau_{2}\leq n\})]-\nonumber\\
-\frac{\alpha(N-1)^{2}}{N^{2}}\sum\limits_{k=1}^{n-1}\frac{k}{N^{2}%
}(1-P\{X_{\tau_{2}}^{(2)}  &  =N^{(x)},\tau_{2}\leq k\})(1-P\{Y_{\tau_{2}%
}^{(2)}=N^{(y)},\tau_{2}\leq k\}).\nonumber
\end{align}
{In order to calculate the expression (\ref{Y2.6}) first we need to find the
conditional probabilities }$P\{X_{\tau_{2}}^{(2)}=N^{(x)},\tau_{2}\leq
k\}${\ and }$P\{Y_{\tau_{2}}^{(2)}=N^{(y)},\tau\leq k\}${\ for every
}$k=1,\ldots,n${ from the Markov sequence (\ref{Y1.4}) with the transitional
probabilities (\ref{Y1.5}). Employing the methods described in
\cite{Y-1a,Y-1b}, it is easy to show that probabilities with the threshold
strategy }$\hat{\tau_{2}(l)\in\mathcal{\mathcal{H}}\text{ are}}${\ }%
\begin{equation}
P\{X_{\tau_{2}}^{(2)}=N^{(x)},\tau_{2}\leq k\}:=P\{X_{\tau_{2}(l)}%
^{(2)}=N^{(x)},\tau_{2}(l)\leq k\}=0 \label{Y2.7}%
\end{equation}
{for }$k=1,\ldots,l-1${\ under the condition of optimal choice, when }%
$\tau_{1}(l)=x_{\hat{\tau}_{1}(l)}^{(1)}=l${, and }$\tau_{2}:=\tau_{2}%
(l)>l\in\mathcal{H}.${\ If }$1\leq k\leq N,${\ then }%
\begin{align}
P\{X_{\tau_{2}(l)}^{(2)}  &  =N^{(x)},\tau_{2}(l)\leq k\}=\sum\limits_{j=l}%
^{k}P\{X_{\tau_{2}(l)}^{(2)}=N^{(x)},\tau_{2}(l)\leq j\}=\nonumber\\
&  =\sum\limits_{j=l}^{k}P\{X_{\tau_{2}(l)}^{(2)}=N^{(x)}|\tau_{2}%
(l)=j\}P\{\tau_{2}(l)=j\}=\nonumber\\
&  =\sum\limits_{j=l}^{k}P\{X_{\tau_{2}(l)}^{(2)}=N^{(x)}|\mathcal{\mathcal{F}%
}_{j}^{(1)}\vee\mathcal{\mathcal{F}}_{j}^{(2)}\}P\{\tau_{2}(l)=j\}=\nonumber\\
&  =\sum\limits_{j=l}^{k}P\{X_{j}^{(2)}=N^{(x)}|\mathcal{\mathcal{F}}%
_{j}^{(1)}\vee\mathcal{\mathcal{F}}_{j}^{(2)}\}P\{\tau_{2}(l)=j\}=\nonumber\\
&  =\sum\limits_{j=l}^{k}\frac{j}{N}P\{\tau_{2}(l)=j\}, \label{Y2.8}%
\end{align}
{and similarly }%
\begin{equation}
P\{Y_{\tau_{2}(l)}^{(2)}=N^{(y)},\tau_{2}(l)\leq k\}=\sum\limits_{j=l}%
^{k}\frac{j}{N}P\{\tau_{2}(l)=j\} \label{Y2.9}%
\end{equation}
{for }$k=l,\ldots,N${. In order to calculate the probabilities }$P\{\tau
_{2}(l)=j\},j\in\mathcal{H},${\ we remark that by Theorem \ref{Tm_Y2.2} in
accordance with the threshold strategy }$\hat{\tau}_{2}(l)\in\mathcal{H}${\ on
the basis of the Kolmogorov equation from the relationships }%
\begin{equation}
P\{x_{\hat{\tau}_{2}}^{(2)}=j\}=\left\{
\begin{array}
[c]{ccc}%
1, &  & j=1,\\
\sum\limits_{i=1}^{j-1}P\{x_{\hat{\tau}_{2}}^{(2)}=i\}p_{ij}^{(1)}; &  & 2\leq
j\leq l-1,\\
\sum\limits_{i=1}^{l-1}P\{x_{\hat{\tau}_{2}}^{(2)}=i\}p_{ij}^{(1)}; &  & 1\leq
j\leq N,
\end{array}
\right.  \label{Y2.10}%
\end{equation}
{one obtains (\ref{Y2.10}) }%
\begin{equation}
P\{x_{\hat{\tau}_{2}}^{(2)}=j\}=\left\{
\begin{array}
[c]{ccc}%
\sum\limits_{k=1}^{j-1}(\mathcal{P}^{k})_{1,j}; &  & 2\leq j\leq l-1,\\
\sum\limits_{s=1}^{l}\sum\limits_{k=1}^{s-1}(\mathcal{P}^{k})_{1,s}%
p_{s,j}^{(2)}; &  & 1\leq j\leq N,
\end{array}
\right.  \label{Y2.11}%
\end{equation}
{where }$\mathcal{P}:=\{p_{ij}^{(2)}:i,j=\overline{0,N}\}${\ is the matrix of
transitional probabilities (\ref{Y1.5}) of the associated Markov process for
the process of choice of the most desirable share package by the first
client-buyer. For convenience, we denote the quantity }$P\{\tau_{2}%
(l)=j\}:=h_{j},j=0,\ldots,N,${\ which is defined by means of (\ref{Y2.11}).
Since }$P\{\tau_{2}(l)=j\}=0${\ for all }$1\leq j\leq l-1${, for the price
function of the optimal choice (\ref{Y2.6}), one obtains }%
\begin{align}
V_{n}^{(1)}  &  =c_{\alpha}[\frac{4n}{N}(1-\sum\limits_{j=l}^{n}\frac{j}%
{N}h_{j})-\frac{n^{2}}{N^{2}}(1-\sum\limits_{j=l}^{n}\frac{j}{N}h_{j}%
)^{2}]-\nonumber\\
&  -\frac{\alpha(N-1)^{2}}{N^{2}}\sum\limits_{k=1}^{n-1}\frac{k}{N^{2}}%
(1-\sum\limits_{j=l}^{n}\frac{j}{N}h_{j})^{2} \label{Y2.12}%
\end{align}
{for all }$n=1,\ldots,\tau_{1}${. Using Theorem \ref{Tm_Y2.2}, let us find the
quantity }$l\in\mathcal{H}${\ satisfying the inequalities (\ref{Y1.6}), which
we can rewrite in the more convenient form: }%
\begin{equation}
V_{l-1}^{(1)}<(\mathcal{P}V^{(1)})_{l-1},(\mathcal{P}V^{(1)})_{l}\leq
V_{l}^{(1)}. \label{Y2.13}%
\end{equation}
{The determining respective inequalities are given by the following analytical
expressions: }%
\begin{equation}%
\begin{array}
[c]{c}%
c_{\alpha}\sum\limits_{k=l}^{N}\frac{(l-1)^{2}[2k(k-1)-1]}{(2l-3)k^{2}%
(k-1)^{2}}\left[  \frac{4k}{N}(1-\sum\limits_{j=l}^{k}\frac{j}{N}h_{j}%
)-\frac{k^{2}}{N^{2}}(1-\sum\limits_{j=l}^{k}\frac{j}{N}h_{j})^{2}\right]  -\\
-\frac{\alpha(N-1)^{2}}{N^{2}}\sum\limits_{k=l}^{N}\frac{(l-1)^{2}%
[2k(k-1)-1]}{(2l-3)k^{2}(k-1)^{2}}\sum\limits_{s=1}^{k-1}\frac{s}{N^{2}%
}(1-\sum\limits_{j=l}^{s}\frac{j}{N}h_{j})^{2}>\\
>c_{\alpha}(\frac{4(l-1)}{N}-\frac{(l-1)^{2}}{N^{2}})-\frac{\alpha(N-1)^{2}%
}{N^{2}}\sum\limits_{k=1}^{l-1}\frac{k}{N^{2}}(1-\sum\limits_{j=l}^{k}\frac
{j}{N}h_{j})^{2}%
\end{array}
\label{Y2.14}%
\end{equation}
{and }%
\begin{equation}%
\begin{array}
[c]{c}%
c_{\alpha}\sum\limits_{k=l}^{N}\frac{l^{2}[2k(k-1)-1]}{(2l-1)k^{2}(k-1)^{2}%
}\left[  \frac{4k}{N}(1-\sum\limits_{j=l}^{k}\frac{j}{N}h_{j})-\frac{k^{2}%
}{N^{2}}(1-\sum\limits_{j=l}^{k}\frac{j}{N}h_{j})^{2}\right]  -\\
-\frac{\alpha(N-1)^{2}}{N^{2}}\sum\limits_{k=l}^{N}\frac{l^{2}[2k(k-1)-1]}%
{(2l-1)k^{2}(k-1)^{2}}\sum\limits_{s=1}^{k-1}\frac{s}{N^{2}}(1-\sum
\limits_{j=l}^{s}\frac{j}{N}h_{j})^{2}\leq\\
\leq c_{\alpha}(\frac{4l}{N}-\frac{l^{2}}{N^{2}})-\frac{\alpha(N-1)^{2}}%
{N^{2}}\sum\limits_{k=1}^{l-1}\frac{k}{N^{2}}(1-\sum\limits_{j=l}^{k}\frac
{j}{N}h_{j})^{2}.
\end{array}
\label{Y2.15}%
\end{equation}

{Let }$l\in\mathcal{H}${\ satisfy the inequalities (\ref{Y2.14}) and
(\ref{Y2.15}). As a result, we obtain the following algebraic equation: }%
\begin{equation}%
\begin{array}
[c]{c}%
\frac{2c_{\alpha}l}{2l-1}\left[  4\frac{l}{N}\ln\frac{N}{l}-\frac{2l^{2}%
}{N^{2}}\ln^{2}\frac{N}{l}-\frac{(N-l)l}{N^{2}}+\frac{2l^{3}}{N^{3}}(\frac
{N}{l}\ln\frac{N}{l}-\frac{N}{l}+1)-\right. \\
-\left.  \frac{l^{4}}{N^{4}}(\frac{N}{l}\ln^{2}\frac{N}{l}-\frac{2N}{l}%
\ln\frac{N}{l}+\frac{N}{l}-1)\right]  -\\
-\alpha\left[  \frac{l^{2}(N-l)}{2N^{3}}+\frac{l(N-l)^{2}}{2N^{3}}-\frac
{l^{2}}{N^{2}}\ln\frac{N}{l}+\frac{l^{2}(N-l)}{N^{3}}+\frac{l^{2}(N-l)^{2}%
}{2N^{4}}+\right. \\
\left.  +\frac{l^{4}}{2N^{4}}(\frac{N}{l}\ln^{2}\frac{N}{l}-\frac{3N}{l}%
\ln\frac{N}{l}+\frac{7N}{2l}-4+\frac{l}{2N})\right]  =\\
=c_{\alpha}(\frac{4l}{N}-\frac{l^{2}}{N^{2}})-\frac{\alpha(N-1)^{2}%
(l-1)l}{2N^{4}},
\end{array}
\label{Y2.16}%
\end{equation}
{where we have taken into account that in accordance with (\ref{Y2.12}),
}$h_{j}=(j-1)p_{1j}^{(2)},${\ }$j=l,\ldots,N${. Then it is straightforward to
verify the following result.}

\begin{theorem}
{\label{Tm_Y2.3} When the procedure of imposing a fine on a buyer is
progressive linear and agrees with the portfolio volume, the Markov sequences
(\ref{Y1.4}) allow the division of the phase space }$\mathcal{H}${\ into the
direct sum of subspaces }$\mathcal{H}_{+}=\{0,\ldots,l-1\}${\ and
}$\mathcal{H}_{-}=\{l,\ldots,N\}${\ under the condition that the promotional
parameter }$c_{\alpha}\geq\alpha/8>0.${\ }
\end{theorem}

{Since the expression (\ref{Y2.16}) is quite complicated when the quantity
}$N\in\mathbb{Z}_{+}${\ is finite, we shall next carry out an asymptotic
analysis under the condition that }$\lim_{N\rightarrow\infty}l(N)/N:=z\in
(0,1),${\ exists, where }$l(N)\in\mathcal{H}${\ is a corresponding solution of
the given equation. }

\subsection{{Asymptotic analysis of the optimal share package choice
strategy}}

{Under the condition that }$\lim_{N\rightarrow\infty}l(N)/N=z\in(0,1)${\ we
obtain \cite{Y-Ge,Y-14} the following transcendental equation from the
algebraic expression (\ref{Y2.16}): }%

\begin{equation}%
\begin{array}
[c]{c}%
\beta(4z\ln z+2z^{2}\ln^{2}z+2z^{2}\ln z+z^{3}\ln^{2}z+2z^{3}\ln
z+5z-z^{3}-z^{4})+\\
+(2z+2z^{2}\ln z+2z^{2}(1-z)^{2}+2z^{3}\ln^{2}z+3z^{3}\ln z+10z^{3}%
-z^{4}+z^{5})=0,
\end{array}
\label{Y3.1}%
\end{equation}
{where we denote }$4c_{\alpha}/\alpha:=\beta\geq1/2${. The transcendental
equation (\ref{Y3.1}) has only one solution on the interval }$(0,1)${, which
can be found by means of numerical methods. }

{The approximate solutions of the equation (\ref{Y3.1}) on the interval
}$(0,1)${\ for some values of }$\beta\in\lbrack0.5,1.5]${\ are shown in the
Table 1. }

{\bigskip}

{Table 1. Real solutions of the equation (\ref{Y3.1}) for different values of
the coefficient }$\beta${. }%
\[
\left\vert
\begin{tabular}
[c]{llllllllllll}\hline
$\beta$ & $\left\vert 0.5\right.  $ & $\left\vert 0.6\right.  $ & $\left\vert
0.7\right.  $ & $\left\vert 0.8\right.  $ & $\left\vert 0.9\right.  $ &
$\left\vert 1.0\right.  $ & $\left\vert 1.1\right.  $ & $\left\vert
1.2\right.  $ & $\left\vert 1.3\right.  $ & $\left\vert 1.4\right.  $ &
$\left\vert 1.5\right.  $\\\hline
$z^{\ast}$ & $\left\vert 0.155\right.  $ & $\left\vert .171\right.  $ &
$\left\vert .186\right.  $ & $\left\vert .199\right.  $ & $\left\vert
.210\right.  $ & $\left\vert .220\right.  $ & $\left\vert .228\right.  $ &
$\left\vert .236\right.  $ & $\left\vert .243\right.  $ & $\left\vert
.249\right.  $ & $\left\vert .254\right.  $\\\hline
\end{tabular}
\ \ \right\vert \
\]
{As a result, we can formulate the following optimal strategy of the
client-buyer behavior in the stock market in terms of bivariant usefulness: At
a large enough bulk }$N\in\mathbb{Z}_{+}${\ of share packages within a bank
portfolio, the first client's optimal behavior strategy for choosing the best
share package is the relative quality value monitoring of }$l=z^{\ast}%
N\in\mathbb{Z}_{+}${\ packages, followed by the choice of the first share
package with bivariant quality surpassing all of the preceding.}

\section{{Concluding Remarks}}

{In contrast to our monovariant profit model, we have used a rather special
discrete Markov process on the phase space }$\mathcal{H}=\{0,1,...,N\}$ to
develop a fairly realistic  simulation in which to formulate an statistically
optimal strategy for choosing the most desirable share package in a zeitnot
market with \ and multiple client-buyers. We showed that {when the number of
share packages in the bank portfolio is sufficiently large, the buyer's
optimal strategy of choice of the most valuable share package is defined by
the universal transcendental equation (\ref{Y3.1}) that depends on the
parameter }$\beta:=4c_{\alpha}/\alpha\geq1/2${, which characterizes the bank
parameter of encouragement and fine (or incentive or disincentive). The loss
risk on the part of the bank, the share-seller, is the lowest when }%
$\beta=1/2${, which leads to the invariant form of the equation (\ref{Y3.1})
with respect to this parameter. In this case, the buyer can skim only }%
$\simeq15.54\%${\ of the share package portfolio to optimally choose the most
valuable share package in the ordered (by desirability) list of packages
following those that are skimmed. }

{It should be also emphasized that when there is no a priori information about
the qualitative characteristics of the portfolios, our statistical model of
zeitnot \ stock behavior of the client-buyers in the bivariant profit function
case is a rather simplified version of the real situation. Moreover, it should
be noted that we have consciously assumed that every client possesses
sufficient financial capital to purchase any bank portfolio share package. If
the client-buyers do not have sufficient resources to buy some of the share
packages for the price offered by the bank, our model and subsequent analysis
would have to be modified. Additional alterations and further development of
our approach would also be required if there is a large number of competing
client-buyers or when there are more profit related parameters affecting the
choice of share packages. But all of these extra degrees of variability are
actually quite typical in large scale banking portfolio markets, so we plan to
investigate these more complex cases in our future research. }

\end{document}